\documentclass[11pt]{article}

\usepackage{amsmath, amssymb, amsthm, amsfonts}
\usepackage{bbm}
\usepackage{hyperref}
\usepackage{ifthen}
\usepackage{fullpage}
\usepackage{tikz}
\usetikzlibrary{positioning,decorations.pathreplacing}

\usepackage{cite}
\usepackage{appendix}
\usepackage{graphicx}
\usepackage{color}
\usepackage[linesnumbered,vlined]{algorithm2e}
\usepackage[noend]{algpseudocode}
\usepackage{epstopdf}
\usepackage{wrapfig}
\usepackage{paralist}
\usepackage[textsize=tiny]{todonotes}

\usepackage{framed}
\usepackage[framemethod=tikz]{mdframed}
\usepackage[bottom]{footmisc}
\usepackage{enumitem}
\setitemize{noitemsep,topsep=3pt,parsep=3pt,partopsep=3pt}
\usepackage[font=small]{caption}
\usepackage{xspace}

\newtheorem{theorem}{Theorem}[section]
\newtheorem{lemma}[theorem]{Lemma}
\newtheorem{meta-theorem}[theorem]{Meta-Theorem}
\newtheorem{claim}[theorem]{Claim}

\newtheorem{corollary}[theorem]{Corollary}

\newtheorem{observation}[theorem]{Observation}

\usepackage{wrapfig}
\usepackage[capitalize,nameinlink]{cleveref}
\crefname{theorem}{Theorem}{Theorems}
\Crefname{lemma}{Lemma}{Lemmas}

\newcommand{\local}{${\mathsf{LOCAL}}$}
\newcommand{\congest}{${\mathsf{CONGEST}}$}

\newcommand{\outdeg}{\mbox{\rm outdeg}}

\newcommand{\dist}{\mbox{\rm dist}}

\newcommand{\SimpleRandColor}{\mathsf{SimpleRandColor}}
\newcommand{\OneShotColoring}{\mathsf{OneShotColoring}}
\newcommand{\ColorBidding}{\mathsf{ColorBidding}}
\newcommand{\DenseColoringStep}{\mathsf{DenseColoringStep}}

\newcommand{\Det}{\mathsf{Det}}

\newcommand{\RecursiveColor}{\mathsf{RecursiveColoring}}
\newcommand{\GeneralColor}{\mathsf{FastColoring}}

\newcommand{\set}[1]{\left\{#1\right\}}

\newcommand{\polylog}{\mathsf{polylog}}
\newcommand{\Exp}{\mathsf{Exp}}
\newcommand{\Prob}{\mathsf{Prob}}
\newcommand{\Pal}{\mathsf{Pal}}

\renewcommand{\paragraph}[1]{\vspace{0.15cm}\noindent {\bf #1}}
\usepackage{wrapfig}

\newcommand{\poly}{\mathsf{poly}}

\begin{document}

%%% Title
%\title{Simulating Any Distributed Algorithm in a Secure Manner}
\title{$(\Delta+1)$ Coloring in the Congested Clique Model\footnote{Department of Computer Science and
    Applied Mathematics, Weizmann Institute of Science, Rehovot 76100,
    Israel. Email:
    \texttt{merav.parter@weizmann.ac.il}.}\footnote{Appeared in ICALP'18.} }
\author{Merav Parter}

\maketitle
\begin{abstract}
In this paper, we present improved algorithms for the $(\Delta+1)$ (vertex) coloring problem in the \emph{Congested Clique} model of distributed computing.
In this model, the input is a graph on $n$ nodes, initially each node knows only its incident edges, and per round each two nodes can exchange $O(\log n)$ bits of information.

Our key result is a randomized $(\Delta+1)$ vertex coloring algorithm that works in $O(\log\log \Delta \cdot \log^* \Delta)$-rounds. This is achieved by combining the recent breakthrough result of [Chang-Li-Pettie, STOC'18] in the \local\ model and a degree reduction technique. We also get the following results with high probability:
(1) $(\Delta+1)$-coloring for $\Delta=O((n/\log n)^{1-\epsilon})$ for any $\epsilon \in (0,1)$, within $O(\log(1/\epsilon)\log^* \Delta)$ rounds, and 
(2) $(\Delta+\Delta^{1/2+o(1)})$-coloring within $O(\log^* \Delta)$ rounds. 
Turning to \emph{deterministic} algorithms, we show a $(\Delta+1)$-coloring algorithm that works in $O(\log \Delta)$ rounds.
\end{abstract}
\newpage
\section{Introduction}
Graph coloring is one of the most central symmetry breaking problems, with a wide range of applications to distributed systems and wireless networks. The most studied coloring problem is the $(\Delta+1)$ \emph{vertex} coloring in which all nodes are given the \emph{same} palette of $\Delta+1$ colors, where $\Delta$ is the maximum degree in the graph. Vertex coloring among other LCL problems\footnote{LCL stands for Locally Checkable Labelling problems, see \cite{naor1995can}.} (e.g., MIS, matching) are traditionally studied in the \local\ model in which any two neighboring vertices in the input graph can exchange arbitrarily long messages. 
\par In recent years there has been a tremendous progress in the understanding of the randomized and the deterministic complexities of many LCL problems in the \local\ model \cite{ChangKP16,ghaffari2016complexity,chang2017time,fischer2017sublogarithmic,newclasses18}. Putting our focus on the $(\Delta+1)$ coloring problem, in a seminal work, Schneider and Wattenhofer \cite{schneider2010new} showed that increasing the number of colors from $\Delta+1$ to $(1+\epsilon)\Delta$ has a dramatic effect on the round complexity and coloring can be computed in just $O(\log^* n)$ rounds when $\epsilon=\Omega(1)$ and $\Delta > \poly\log n$. This has led to two recent breakthroughs. 
Harris, Schneider and Su \cite{harris2016distributed} showed an $O(\sqrt{\log \Delta})$-round algorithm for $(\Delta+1)$ coloring, providing a separation for the first time between MIS and coloring (due to the MIS lower bound of \cite{kuhn2016local}). In a recent follow-up breakthrough, Chang, Li and Pettie \cite{CHP18} extended the technique of \cite{harris2016distributed} to obtain the remarkable and quite extraordinary round complexity of $O(\log^* n+\Det_{\deg}(\polylog n))$ for the $(\Delta+1)$-list coloring problem where $\Det_{\deg}(n')$ is the deterministic round complexity of  $(\deg+1)$ list coloring algorithm\footnote{In the $(\deg+1)$ list coloring problem, each vertex $v$ is given a palette with $\deg(v,G)+1$ colors.} in $n'$-vertex graph. 
Both of these recent breakthroughs use messages of large size, potentially of $\Omega(n)$ bits. 

In view of these recent advances, the understanding of LCL problems in bandwidth-restricted models is much more lacking. 
Among these models, the congested clique model \cite{lotker2005minimum}, which allows all-to-all communication has attracted a lot of attention in the last decade and more recently, in the context of LCL problems \cite{berns2012super,hegeman2014near,hegeman2015lessons,Censor-HillelPS17,GhaffariMIS17,gregorylocal18}. In the congested clique model, each node can send $O(\log n)$ bits of information to any node in the network (i.e., even if they are not connected in the input graph). 
The ubiquitous of overlay networks and large scale distributed networks make the congested clique model far more relevant (compared to the \local\ and the \congest\ models) in certain settings.
 
\paragraph{Randomized LCL in the Congested Clique Model.}
Starting with Barenboim et al. \cite{barenboim2016locality}, currently, all efficient randomized algorithms for classical LCL problems have the following structure: an initial randomized phase and a \emph{post-shattering} deterministic phase. The shattering effect of the randomized phase which dates back to Beck \cite{beck1991algorithmic}, breaks the graph into subproblems of $\poly \log n$ size to be solved deterministically. In the congested-clique model, the shattering effect has an even more dramatic effect. Usually, a node survives (i.e., remained undecided) the randomized phase with probability of $1/\poly(\Delta)$. Hence, in expectation the size of the remaining unsolved graph is\footnote{Using the bounded dependencies between decisions, this holds also with high probability.} $O(n)$. At that point, the entire unsolved subgraph can be solved in $O(1)$ rounds, using standard congested clique tools (e.g., the routing algorithm by Lenzen \cite{lenzen2013route}). Thus, as long as the main randomized part uses short messages, the congested clique model ``immediately" enjoys an improved round complexity compared to that of the \local\ model.  
 
In a recent work \cite{GhaffariMIS17}, Ghaffari took it few steps farther and showed an $\widetilde{O}(\sqrt{\log \Delta})$-round randomized algorithm for MIS in the congested clique model, improving upon the state-of-the-art complexity of $O(\log \Delta+2^{O(\sqrt{\log\log n})})$ rounds in the \local\ model, also by Ghaffari \cite{ghaffari2016improved}. 
When considering the $(\Delta+1)$ coloring problem, the picture is somewhat puzzling. On the one hand, in the \local\ model, $(\Delta+1)$ coloring is provably simpler then MIS. However, since all existing $o(\log \Delta)$-round algorithms for $(\Delta+1)$ coloring in the \local\ model, use large messages, it is not even clear if the power of all-to-all communication in the congested clique model can compensate for its bandwidth limitation and outperform the \local\ round complexity, not to say, even just match it.
We note that on hind-sight, the situation for MIS in the congested clique was somewhat more hopeful (compared to coloring), for the following reason. The randomized phase of Ghaffari's MIS algorithm although being in the \local\ model \cite{ghaffari2016improved}, used \emph{small} messages and hence could be implemented in the \congest\ model with the same round complexity. 
%This is in strike contrast to coloring, where all existing $o(\log \Delta)$-round algorithms\footnote{The main parts of these algorithm, i.e., even without the post-shattering phase.} in the \local\ model could not be simulated as is, in the congested clique model. 
To sum up, currently, there is no $o(\log \Delta)$-round algorithm for $(\Delta+1)$ coloring in any bandwidth restricted model, not even in the congested-clique.
%In this paper we answer these problems in the affirmative and present an $O(\log\log \Delta)$

%
%\paragraph{Randomized vs. Deterministic LCL in \local\ Model.}
%One fundamental characteristic of distributed algorithms for local problems is whether they are deterministic or randomized. Currently, there exists a curious gap between the known complexities of randomized and deterministic solutions for local problems. Interestingly, the main indistinguishability-based technique used for obtaining the relatively few lower bounds that are known seems unsuitable for separating these cases.
%A beautiful recent work of Chang et al.~\cite{ChangKP16} sheds some light over this, by proving that the randomized complexity of any local problem is at least its deterministic complexity on instances of size $\sqrt{\log n}$. In addition, building upon a new lower bound technique of Brandt et al.~\cite{BrandtFHKLRSU16}, they show an exponential separation between the randomized and deterministic complexity of $\Delta$-coloring trees. These results hold in the \local\ model, which allows unbounded messages. Recently, \cite{ghaffari2017derandomizing} designed a general dernadomiation technique for LCL problems in the \local\ model (indeed, the messages of their algorithm might be polynomially large.)

\paragraph{Derandomization of LCL in the Congested-Clique Model.}
There exists a curious gap between the known complexities of randomized and deterministic solutions for local problems in the \local\ model (\cite{ChangKP16,ghaffari2016complexity}).
Censor et al. \cite{Censor-HillelPS17} initiated the study of \emph{deterministic} LCL algorithms in the congested clique model by means of derandomization. The main take home message of \cite{Censor-HillelPS17} is as follows: for most of the classical LCL problems there are $\poly \log n$ round randomized  algorithms (even in the \congest\ model). For these algorithms, it is usually sufficient that the random choices made by vertices are \emph{almost} independent.  This implies that each round of the randomized algorithm can be simulated by giving all nodes a shared random seed of  $\poly \log n$ bits. To dernadomize a single round of the randomized algorithm, nodes should compute (deterministically) a seed which is at least as ``good"\footnote{The random seed is usually shown provide a large progress in expectation. The deterministically computed seed should provide a progress at least as large as the expected progress of a random seed.} as a random seed would be. To compute this seed, they need to estimate their ``local progress" when simulating the random choices using that seed. %The lauch an estimation usually depends on their immediate neighborhood. 
%Due to the bounded bandwidth, it is sometimes impossible to compute the exact progress guaranteed by the seed. Instead the vertices compute a \emph{pessimistic estimator}, which provides a lower bound for the true measure and has the benefit that it requires much less information from the local neighborhood, and thus can be computed efficiently even with bandwidth limitation. 
Combining the techniques of conditional expectation, pessimistic estimators and bounded independence leads to a simple ``voting"-like algorithm in which the bits of the seed are computed 
\emph{bit-by-bit}. Once all bits of the seed are computed, it is used to simulate the random choices of that round. %In the current work, we reveal an $O(\log n)$-size seed in $O(1)$ many rounds. 
%Such a scheme also works in the more restricted \emph{Broadcast Congested Clique} model, in which a node must send the same $O(\log n)$-bit message to all other nodes in any single round. This approach has be shown to yield a deterministic MIS algorithm that works in $O(\log^2 \Delta)$ rounds and a deterministic construction of $(2k-1)$ spanner in $O(k \log n)$ rounds. 
For a recent work on other complexity aspects in the congested clique, see \cite{korhonen2017brief}. 

%\paragraph{Recent breakthroughs in vertex coloring algorithm.}
%
%difference between MIS and coloring
%
%\paragraph{Randomized LCL in the congested clique model}
%\cite{hegeman2015lessons}
%
%Simple $O(\Delta^2)$ coloring in $O(1)$ rounds. 
%\cite{GhaffariMIS17}

\subsection{Main Results and Our Approach}
We show that the power of all-to-all communication compensates for the bandwidth restriction of the model:
%
%$(\Delta+1)$ vertex coloring can indeed be solved significantly faster (e.g., compared to the distributed models of \local):
\begin{mdframed}[hidealllines=true,backgroundcolor=gray!25]
\vspace{-8pt}
\begin{theorem}\label{thm:maincol} There is a randomized algorithm that computes a $(\Delta+1)$ coloring in $O(\log\log \Delta \cdot \log^* n)$ rounds of the congested clique model, with high probability\footnote{As usual, by high probability we mean $1-1/n^c$ for some constant $c\geq 1$.}. 
\vspace{-3pt}
\end{theorem}
\end{mdframed}
This significantly improves over the state-of-the-art of $O(\log \Delta)$-round  algorithm for $(\Delta+1)$ in the congested clique model. It should also be compared with the round complexity of $(2^{O(\sqrt{\log\log n})})$ in the \local\ model, due to \cite{CHP18}.  As noted by the authors, reducing the \local\ complexity to below $O((\log\log n)^2)$ requires a radically new approach.

Our $O(\log\log \Delta \cdot \log^*n)$ round algorithm is based on a recursive degree reduction technique which can be used to color any almost-clique graph with $\Delta=\widetilde{O}(n^{1-o(1)})$ in essentially $O(\log^* \Delta)$ rounds. 
\begin{mdframed}[hidealllines=true,backgroundcolor=gray!25]
\vspace{-8pt}
\begin{theorem}\label{thm:coleps} (i) For every $\epsilon \in (0,1)$, there is a randomized algorithm that computes a $(\Delta+1)$ coloring in $O(\log(1/\epsilon) \cdot \log^* n)$ rounds for graphs with $\Delta=O((n/\log n)^{1-\epsilon})$, (ii) This also yields a $(\Delta+\Delta^{1/2+o(1)})$ coloring in $O(\log^*\Delta)$ rounds, with high probability. 
\vspace{-3pt}
\end{theorem}
\end{mdframed}
Claim (ii) improves over the $O(\Delta)$-coloring algorithm of \cite{hegeman2015lessons} that takes $O(\log\log \log n)$ rounds in \emph{expectation}. 
We also provide fast \emph{deterministic} algorithms for $(\Delta+1)$ list coloring. The stat-of-the-art in the \local\ model is $\widetilde{O}(\sqrt{\Delta})+\log^* n$ rounds due to Fraigniaud, Heinrich, Marc and Kosowski \cite{fraigniaud2016local}. 
\begin{mdframed}[hidealllines=true,backgroundcolor=gray!25]
\vspace{-8pt}
\begin{theorem}\label{thm:coldet}
There is a deterministic algorithm that computes a $(\Delta+1)$ coloring in $O(\log \Delta)$ rounds of the congested clique model and an $O(\Delta^2)$ coloring in $O(1)$ rounds.
\vspace{-3pt}
\end{theorem}
\end{mdframed}
In \cite{Censor-HillelPS17}, a deterministic algorithm for $(\Delta+1)$ coloring in $O(\log \Delta)$ rounds was shown only for the case where $\Delta=O(n^{1/3})$. Here it is extended for $\Delta=\Omega(n^{1/3})$. This is done by derandomizing an $(\Delta+1)$-list coloring algorithm which runs in $O(\log n)$ rounds. Similarly to \cite{Censor-HillelPS17}, we first show that this algorithm can be simulated when the random choices made by the nodes are pairwise independent. Then, we enjoy the small search space and employ the method of conditional expectations. Instead of computing the seed bit by bit, we compute it in chunks of $\lfloor \log n\rfloor$ bits at a time, by fully exploiting the all-to-all power of the model.

\paragraph{The Challenges and the Degree Reduction Technique.} 
Our starting observation is that the CLP algorithm \cite{CHP18} can be implemented in $O(\log^* \Delta)$ rounds in congested clique model for $\Delta=O(\sqrt{n})$. When $\Delta=O(\sqrt{n})$, using Lenzen's routing algorithm \cite{lenzen2013route}, each node can learn in $O(1)$ rounds, the palettes of all its neighbors along with the neighbors of its neighbors. Such knowledge is mostly sufficient for the CLP algorithm to go through. 
%Extending the CLP algorithm in the congested clique model for $\Delta=\Omega(\sqrt{n})$ seems to be a non-trivial mission for several reasons that we discuss in the technical sections. 

%\\ \\
%\textbf{*****************************************************************}
%We note that combining the random partitioning with the CLP algorithm (adopted to the congested clique) yields a simple $O(\log^* n)$-round algorithm that colors every graph with $\Delta+\widetilde{O}(\Delta^{3/4})$ colors: partition $G$ into $\Delta^{1/2}$ subgraphs by letting each vertex pick a subgraph uniformly at random. Vy simple application of Chernoff, the maximum degree in each $G_i$ is $\Delta(G_i)\leq \sqrt{\Delta}+\widetilde{O}(\Delta^{1/4})=O(\sqrt{n})$. Allocate $\Delta(G_i)+1$ colors to each subgraph and the apply the CLP algorithm in each of the $G_i$ subgraphs simultaneously. Overall, the number of colors allocated is $\Delta+\widetilde{O}(\Delta^{3/4})$. 
%\textbf{*****************************************************************}
%\\ \\
To handle large degree graphs, we design a graph sparsification technique that essentially reduces the problem of $(\Delta+1)$ coloring for an arbitrarily large $\Delta=\widetilde{O}(n^{1-\epsilon})$ into $\ell=O(\log(1/\epsilon))$ (non-independent) subproblems. In each subproblem, one has to compute a $(\Delta'+1)$ coloring for a subgraph with $\Delta'=O(\sqrt{n})$, which can be done in $O(\log^* \Delta)$ rounds, using a modification of the CLP algorithm, that we describe later on. 

\setlength{\columnsep}{18pt}%
\begin{wrapfigure}{r}{7cm}
\centering
\vspace{-5pt}
\includegraphics[width=0.28\textwidth]{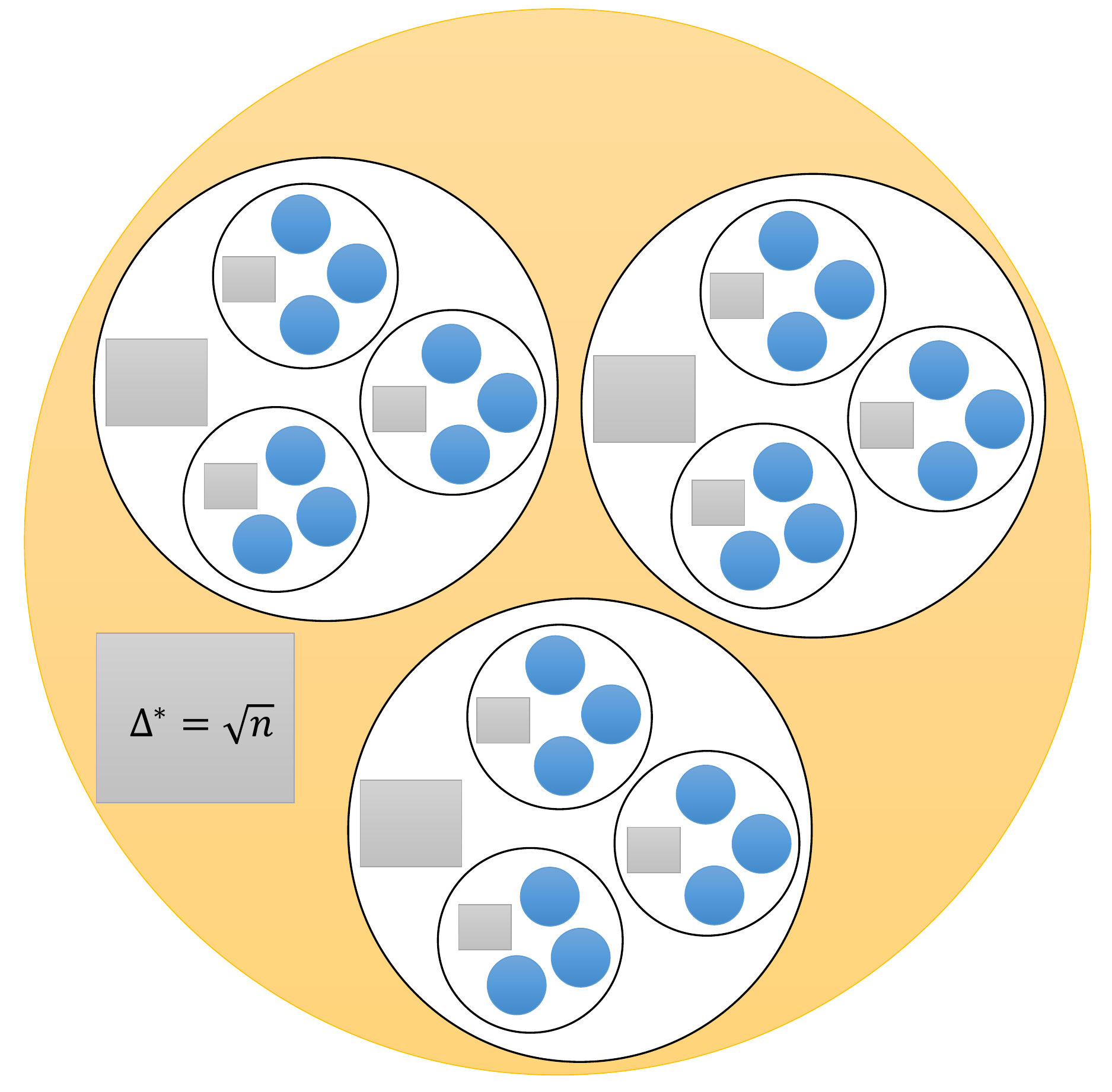}
\vspace{-5pt}
	\caption{{\footnotesize Illustration of the recursive sparsification. Gray boxes correspond to subgraphs with maximum degree $O(\sqrt{\Delta})$.}
		\vspace{-2pt}}
	\label{fig:recursivecol}
\end{wrapfigure} 
Since there many dependencies between these $\ell$ 
sub-problems, it is required by our algorithm to solve them one-by-one, leading to a round complexity of $O(\log(1/\epsilon) \log^* \Delta)$.  See \Cref{fig:recursivecol} for an illustration of the recursion levels. 
To get an intuition into our approach and the challenges involved, 
consider an input graph $G$ with maximum degree $\Delta=(n/\log n)^{1-\epsilon}$ and a palette $\Pal(G)=\{1,\ldots, \Delta+1\}$  given to each node in $G$. 
A natural approach (also taken in \cite{hegeman2015lessons}) for handling a large degree graph is to decompose  it (say, randomly) into $k$ vertex disjoint graphs $G_1, G_2, \ldots , G_{k}$, allocate a \emph{distinct} set of colors for each of the subgraphs taken from $\Pal(G)$ and solve the problem recursively on each of them, enjoying (hopefully) smaller degrees in each $G_i$. 
%That is, the algorithm allocates each $G_i$ a palette $\Pal(G_i) \subseteq \Pal(G)$ of $\Delta(G_i)+1$ colors such that any two subgraphs $G_i,G_j$ are given disjoint set of palettes $\Pal(G_i)\cap \Pal(G_j)=\emptyset$. 
Intuitively, assigning a \emph{disjoint} set of colors to each $G_i$ has the effect of \emph{removing} all edges connecting nodes in different subgraphs. Thus, the input graph $G$ is sparsified into a graph $G'=\bigcup G_i$ such that a legal coloring of $G'$ (with the corresponding palettes given to the nodes) is a legal coloring for $G$. The main obstacle in implementing this approach is that assigning a distinct set of $\Delta(G_i)+1$ colors to each of the $G_i$ subgraph might be beyond the budget of $\Delta+1$ colors. Indeed in \cite{hegeman2015lessons} this approach led to $O(\Delta)$ coloring rather than $(\Delta+1)$. To reduce the number of colors allocated to each subgraph $G_i$, it is desirable that the maximum degree $\Delta(G_i)$ would be as small as possible, for each $G_i$. This is exactly the problem of $(k,p)$ \emph{defective coloring} where one needs to color the graph with $k$ colors such that the number of neighbors with the same color is at most $p$. To this point, the best defective coloring algorithm for large degrees is the randomized one: let each node pick a subgraph $G_i$ (i.e., a color in the defective coloring language) uniformly at random. By a simple application of Chernoff bound, it is easy to see that the partitioning is ``almost" perfect: w.h.p., for every $i$, $\Delta(G_i)\leq \Delta/k+\sqrt{\log n\cdot \Delta/k}$. Hence, allocating $\Delta(G_i)+1$ colors to each subgraphs consumes $\Delta+\widetilde{O}(\sqrt{\Delta k})$ colors. To add insult to injury, this additive penalty of $\widetilde{O}(\sqrt{\Delta k})$ is only for one recursion call! 

It is interesting to note that the parameter $k$ -- number of subgraphs (colors) -- plays a key role here. Having a large $k$ has the benefit of sharply decreasing the degree (i.e., from $\Delta$ to $\Delta/k$). However, it has the drawback of increasing the standard deviation and hence the total number  of colors used. Despite these opposing effects, it seems that for whatever value of $k$ chosen, increasing the number of colors to $\Delta+\Delta^{\epsilon}$ is unavoidable.

Our approach bypasses this obstacle by partitioning only a large fraction of the vertices into small-degree subgraphs \emph{but not all of them}. Keeping in mind that we can handle efficiently graphs with maximum degree $\sqrt{n}$, in every level of the recursion, roughly $1-1/\sqrt{\Delta}$ of the vertices are partitioned into subgraphs $G_1,\ldots, G_k$. Let $\Delta(G_i)$ be the maximum degree of $G_i$. The remaining vertices join a left-over subgraph $G^*$. The number of subgraphs, $k$, is chosen carefully so that allocating $\Delta(G_i)+1$ colors to each of the $k$ subgraphs, consumes at most $\Delta$ colors, on the one hand; and that the degree reduction in each recursion level is large enough on the other hand. These subgraphs are then colored recursively, until all remaining subgraphs have degree of $O(\sqrt{n})$. Once all vertices in these subgraphs are colored, the algorithm turns to color the left-over subgraph $G^*$. Since the maximum degree in $G^*$ is $O(\sqrt{n})$, it is tempting to use the CLP algorithm to complete the coloring, as this can be done in $O(\log^* n)$ rounds for such bound on the maximum degree. This is not so immediate for the following reasons. Although the degree of $v$ in $G^*$ is $O(\sqrt{n})$, the graph $G^*$ cannot be colored independently (as at that point, we ran out of colors to be solely allocate to $G^*$). Instead, the coloring of $G^*$ should agree with the coloring of the rest of the graph and each $v$ might have $\Omega(\Delta)=n^{1-\epsilon}$ neighbors in $G$. At first glance, it seems that this obstacle is easily solved by letting each $v \in G^*$ pick a subset of $\deg(v,G^*)=O(\sqrt{n})$ colors from its palette (i.e., removing the colors taken by its neighbors in $G\setminus G^*$). Now, one can consider only the graph $G^*$ with maximum degree $\sqrt{n}$, where each vertex has a palette of $\deg(v,G^*)+1$ colors. Unfortunately, this seemingly plausible approach has a subtle flaw: for the CLP algorithm it is essential that each vertex receives a palette with \emph{exactly} $\Delta(G^*)+1$ colors. This is indeed crucial and as noted by the authors adopting their algorithm to a $(\deg+1)$ coloring algorithm is highly non-trivial and probably calls for a different approach.

In our setting, allocating each vertex $v \in G$ the exact same number of colors seems to be impossible as the number of available colors of each $v$ depends on the number of its neighbors in $G\setminus G^*$, and this number has some fluctuations due to the random partitioning of the vertices. 
To get out of this impasse, we show that after coloring all vertices in $G\setminus G^*$, every vertex $v \in G^*$ has $r_v  \in [\Delta(G^*) \pm (\Delta(G^*))^{3/5}]$ available colors in its palette where $r_v \geq \deg(v,G^*)$. 
In other words, all vertices can be allocated ``almost" the same number of colors, but not exactly the same. We then carefully revise the basic definitions of the CLP algorithm and show that the analysis still goes through (upon minor changes) for this narrow range of variation in the size of the palettes. 

%\paragraph{Some Remaining Open Problems.}
%We note that our approach (as well as that of \cite{hegeman2015lessons}) works only for the basic problem of $(\Delta+1)$ coloring and does not extend for the related problem of $(\Delta+1)$ \emph{list} coloring. In the latter, each node in $G$ might have a different palette of $\Delta+1$ colors whereas in the former all nodes are given the same initial palette with $\Delta+1$ colors. Our deterministic algorithms do work for the $(\Delta+1)$ list coloring. We leave open the question of designing $(\Delta+1)$ \emph{list} coloring algorithms in the congested clique with $o(\log \Delta)$ rounds for $\Delta=\Omega(\sqrt{n})$. 
%Turning back to the task of $(\Delta+1)$ coloring, the only barrier to achieve an $O(\log^* n)$ round for any graph $G$ boils into computing defective coloring \emph{slightly} better then the randomized approach (i.e., decreasing the standard deviation). We leave open the following problem: Compute $(k,p)$ defective coloring for $p \leq \Delta/k+(\Delta/k)^{1/4}$ for $\Delta=\Omega(\sqrt{n})$ in $O(\log^* n)$ rounds in the congested clique model.
\paragraph{Paper Organization.}
In \Cref{sec:clpcc}, we explain how the CLP algorithm of \cite{CHP18} can be simulated in $O(\log^* \Delta)$ congested-clique rounds when $\Delta=O(\sqrt{n})$. In \Cref{sec:deltathreefour}, we illustrate the degree-reduction technique on the case where $\Delta=O((n/\log n)^{3/4})$. \Cref{sec:deltaeps} extends this approach for 
$\Delta=O((n/\log n)^{1-\epsilon})$ for any $\epsilon \in (0,1)$, and \Cref{sec:generalcol} handles the general case and provides the complete algorithm. Finally, \Cref{sec:detcol} discusses deterministic coloring algorithms. 

\section{The Chang-Li-Pettie (CLP) Alg. in the Congested Clique}\label{sec:clpcc}

\paragraph{High-level Description of the CLP Alg. in the \local\ Model.}
In the description below, we focus on the main randomized part of the CLP algorithm \cite{harris2016distributed}, and start by providing key definitions and ideas from \cite{harris2016distributed}. %We start by describing the main part of Harris-Schneider-Su algorithm \cite{harris2016distributed}. 
%Schneider and Wattenhofer showed an $O(\log^* n)$-round algorithm when the palette size is $(1+\epsilon)\Delta$ for constant $\epsilon \in [0,1]$. 

Harris-Schneider-Su algorithm is based on partitioning the graph into an $\epsilon$-sparse subgraph and a collection of vertex-disjoint $\epsilon$-dense components, for a given input parameter $\epsilon$. 
Since the CLP algorithm extends this partitioning, we next formally provide the basic definitions from \cite{harris2016distributed}.
For an $\epsilon \in (0,1)$, an edge $e=(u,v)$ is an $\epsilon$-friend if $|N(u) \cap N(v)|\geq (1-\epsilon)\cdot \Delta$. The endpoints of an $\epsilon$-friend edge are $\epsilon$-friends.
A vertex $v$ is $\epsilon$-dense if $v$ has at least $(1-\epsilon)\Delta$ $\epsilon$-friends, otherwise it is $\epsilon$-sparse. %When $\Delta=O(\sqrt{n})$, each vertex can collect the neighbors of its neighbors in $O(1)$ rounds, compute its $\epsilon$-friends, and decide if it is an $\epsilon$-dense or $\epsilon$-sparse. 
A key structure that arises from the definition of $\epsilon$ dense vertices is that of $\epsilon$-\emph{almost clique} which is a connected component of the subgraph induced by the $\epsilon$-dense vertices and $\epsilon$-friend edges. 
The dense components, $\epsilon$-\emph{almost} cliques, have some nice properties: each component $C$ has at most $(1+\epsilon)\Delta$ many vertices, each vertex $v \in C$ has $O(\epsilon \Delta)$ neighbors outside $C$ (called \emph{external} neighbors) and $O(\epsilon\Delta)$ vertices in $C$ which are not its neighbors. In addition, $C$ has weak diameter at most $2$. Coloring the dense vertices consists of $O(\log_{1/\epsilon} \Delta)$ phases. 
The efficient coloring of dense regions is made possible by generating a random proper coloring inside each clique so that each vertex has a small probability of receiving the same color as one of its external neighbors. 
To do that, in each cluster a random permutation is computed and each vertex selects a tentative color from its palette excluding the colors selected by
lower rank vertices. Since each component has weak diameter at most $2$, this process is implemented in $2$ rounds of the \local\ model. The remaining \emph{sparse} subgraph is colored using a Schneider-Wattenhofer style algorithm \cite{schneider2010new} within $O(\log(1/\epsilon))$ rounds.

In Chang-Li-Pettie algorithm the vertices are partitioned into $\ell=\lceil \log\log \Delta \rceil$ layers in decreasing level of density. This hierarchical partitioning is based on a sequence of $\ell$ sparsity thresholds $\epsilon_1,\ldots,\epsilon_\ell$ where $\epsilon_i=\sqrt{\epsilon_{i-1}}$. Roughly speaking, level $i$ consists of the vertices which are $\epsilon_i$-dense but $\epsilon_{i-1}$-sparse.
Instead of coloring the vertices layer by layer, the algorithm partitions the vertices in level $i$ into large and small components and partitions the layers into $O(\log^* \Delta)$ strata, vertices in the same stratum would be colored simultaneously.  The algorithm colors vertices in $O(\log^* \Delta)$ phases, giving priority to vertices in small components. The procedures that color the dense vertices are of the same flavor as those of Harris-Schneider-Su. The key benefit in having the hierarchical structure is that the dense-coloring procedure is applied for $O(1)$ many phases on each stratum, rather than applying it for $O(\log_{1/\epsilon}\Delta)$ phases as in \cite{harris2016distributed}. 

\paragraph{An $O(\log^* \Delta)$-Round Alg. for $\Delta=O(\sqrt{n})$ in the Congested Clique.}
We next observe that the randomized part of the CLP algorithm \cite{CHP18} can be implemented in the congested clique model when $\Delta=O(\sqrt{n})$ within $O(\log^* \Delta)$ rounds. We note that we obtain a round complexity of $O(\log \Delta^*)$ rather than $O(\log^* n)$ as in \cite{CHP18}, due to the fact that the only part of the CLP algorithm that requires $O(\log^* n)$ rounds was for coloring a subgraph with maximum constant degree. In the congested-clique model such a step can be implemented in $O(1)$ rounds using Lenzen's routing algorithm. We show:
\begin{theorem}\label{thm:colorsqrtn}
For every graph with maximum degree $\Delta=O(\sqrt{n})$, there is an $O(\log^* \Delta)$-round randomized algorithm that computes $(\Delta+1)$-list coloring in the congested clique model.
\end{theorem}
The main advantage of having small degrees is that it is possible
for each node to collect its $2$-neighborhood in $O(1)$ rounds (i.e., using Lenzen's routing \cite{lenzen2013route}). As we will see, this is sufficient in order to simulate the CLP algorithm in $O(\log^* \Delta)$ rounds. The hierarchical decomposition of the vertices depends on the computation of $\epsilon$-dense vertices. By collecting the neighbors of its neighbors, every vertex can learn its $\epsilon$-dense friends and based on that deduce if it is an $\epsilon$-dense vertex for every $\epsilon$. 
In particular, for every edge $(u,v)$, $v$ can learn the minimum $i$ such that $u$ and $v$ are $\epsilon_i$-friends (i.e, the ``threshold" $\epsilon$). To allow each vertex $v$ compute the $\epsilon$-almost cliques to which it belongs, we do as follows. Each vertex $v$ sends to each of its neighbors $N(v)$, the minimum $\epsilon_i$ such that $u,v$ are $\epsilon_i$-friends, for every $u \in N(v)$. Since the weak diameter of each almost-clique is at most $2$, each vertex has collected all the required information from its $2^{nd}$ neighborhood to locally compute its $\epsilon_i$-almost cliques for every $\epsilon_i$. Overall, each vertex sends $O(\Delta)$ messages and receives $O(\Delta^2)=O(n)$ messages, collection this information can be done in $O(1)$ rounds for all nodes, using Lenzen's routing algorithm. 

The next obstacle is the simulation of the algorithm that colors the $\epsilon$-dense vertices. 
Since each $\epsilon$-almost clique $C$ has $(1+\epsilon)\Delta=O(\sqrt{n})$ vertices, we can make the leader of each such $C$ learn the palettes of all the vertices in its clique as well as their neighbors in $O(1)$ rounds. The leader can then locally simulate the dense-coloring procedure and notify the output color to each of its almost-clique vertices. Finally, coloring the sparse regions in a Schneider-Wattenhofer style uses messages of size $O(\Delta)$ and hence each vertex is the target of $O(\Delta^2)=O(n)$ messages which again can be implemented in $O(1)$ many rounds. 
By the above description, we also have: 
\begin{corollary}\label{cor:colorsqrtnparallel}
Given $q$ \emph{vertex}-disjoint subgraphs $G_1,\ldots, G_q$ each with maximum degree $\Delta=O(\sqrt{n})$, a $(\Delta+1)$ coloring can be computed in $O(\log^* \Delta)$ rounds, for all subgraphs simultaneously.
\end{corollary}
A more detailed description of the algorithm and the proof of Cor. \ref{cor:colorsqrtnparallel} appears in \Cref{sec:clpdetails}.
\def\APPENDFULLCLPSMALL{
The hierarchy is based on a sequence of $\ell=O(\log\log \Delta)$ sparsity parameters $(\epsilon_1,\ldots, \epsilon_\ell)$ and a set of vertices $V^*$ that are not yet colored. The sequence of sparsity parameters satisfies: (i) $\epsilon_1=\Delta^{-1/10}$, $\epsilon_i=\sqrt{\epsilon_{i-1}}$ and $\epsilon_\ell=1/K$ for a large enough constant $K$. 

For a sparsity parameter $\epsilon_i$, let $V^d_{\epsilon_i},V^{s}_{\epsilon_i}$ be the set of vertices which are $\epsilon_i$-dense (resp., sparse). Based on its $2$-hop neighborhood, each vertex locally computes if it is $\epsilon_i$-dense (in $V^d_{\epsilon_i}$) or $\epsilon_i$-sparse. This $\ell=O(\log\log \Delta)$ bits of information (i.e., there are $\ell$ sparsity parameters) are exchanged between neighbors on $G$ in a single round\footnote{In fact, it is even sufficient to send the first $i$ such that $v$ is an $\epsilon_i$-dense.}.
%Each vertex can send to all its neighbors the outcome of that computation which consists of $O(\log\log \Delta)$ bits 
%(i.e., for each $\epsilon_i$, $v$ sends $1$ if its $\epsilon_i$ dense and $0$ otherwise)
The sparsity of the vertices define a hierarchy of $\ell$ \emph{levels}: $V_1, \ldots, V_\epsilon$ where $V_1=V^* \cap V_{\epsilon_1}^d$, $V_i=V^* \cap (V_{\epsilon_i}^d \setminus V_{\epsilon_{i-1}}^d)$ and $V_{sp}=V^* \cap V_{\epsilon_\ell}$. These $\ell$ levels are further grouped 
$s=O(\log^* \Delta)$ \emph{Strata} $W_1,\ldots, W_s$ where $W_1=V_1$ and 
$$W_k=\bigcup_{i: \epsilon_i \in (\xi_{k-1},\xi_k]} V_i \mbox{~~where~} \zeta_1=\epsilon_1, ~, \xi_k=1/\log(1/\xi_{k-1})~.$$ 

The computation of $\epsilon_i$-\emph{almost} clique for every $\epsilon_i$ can be computed
using the $O(1)$-round connectivity algorithm of \cite{Jurdzinski018}. That is, the $\epsilon_i$-\emph{almost} clique are the connected components of the graph $G[V_{\epsilon_i}^d]$. Thus, in $O(\log\log \Delta)$ rounds, each vertex knows the members of its $\epsilon_i$-almost clique for every $\epsilon_i$.
 
Each $i$-layer $V_i$ is partitioned into \emph{blocks} based on these clique structures. In particular, letting $\{C_1, C_2, \ldots, \}$ be $\epsilon_i$-almost cliques, then each clique $C_j$ defines a block $B_j=C_j \cap V_i$, that is the block $B_j$ contains the subset of vertices in $C_j$ that are $\epsilon_i$-dense but are \emph{not} $\epsilon_{i-1}$-dense. These blocks can be easily computed in $O(1)$ rounds as they only depend on the layer of the vertex and on the knowledge of the cliques. 
This block collection $(B_{i,1},B_{i,2}, \ldots, )$ is a partition of $V_i$, we call a block $B_{i,j}$ as a layer-$i$ block (and also stratum $k$ block for the right $k$).
The algorithm uses a tree structure on these blocks: A layer-$i$ block $B$ is \emph{descendant} of layer-$j$ block $B'$ if $j>i$ and both $B, B'$ are subsets of the same $\epsilon_{j}$-almost clique. The root contains of the subset of sparse vertices, $V_{\epsilon_\ell}$. In $O(1)$ rounds, each vertex $v$ knows for each vertex $u \in G$ the first index $i$ such that $u$ is $\epsilon_i$-dense. In addition, in $O(\log\log\Delta)$, the vertex can tell the $\epsilon_i$-almost clique ID of all the vertices in the graph, for every $i \in \{1,\ldots, \ell\}$. Thus the entire block tree $\mathcal{T}$ is known to all the vertices.

The algorithm colors vertices according to their stratum. The vertices in the stratum are divided into two: those belonging to \emph{large} blocks and those remaining to \emph{small} blocks.
A stratum-$k$ block $B$ is a \emph{large} block if $|B|\geq \Delta/\log^2(1/\xi_k)$ and there is no other 
stratum-$k'$ block for $k'\geq k$ such that $B$ is a descendent of $B'$ and $|B'|\geq \Delta/\log^2(1/\xi_{k'})$. Otherwise, $B$ is \emph{small}. Define $V_i^S,V_i^L,W_k^S$ and $W_k^L$ be the set of all vertices in layer-$i$ small blocks, layer-$i$ large blocks, stratum-$k$ small blocks, and stratum-$k$ large blocks.
Finally, the algorithm defines a partition of $W_k$ into super-blocks $(R_1,R_2, \ldots, )$ in the following manner. Let $i',i'+1,\ldots, i$ be the set of layers of stratum $k$. Let $C_1,\ldots, C_i$ be the set of $\epsilon_i$-almost cliques. Then for each $C_j$, define $R_j=C_j \cap W_k$ as the \emph{stratum}-$k$ super-block.

\paragraph{Main Steps of the Coloring algorithm.}
The first step of the algorithm applies Alg. $\OneShotColoring$ for $O(1)$ rounds which colors a small constant fraction of the vertices. In this algorithm the vertices simply send one color to their neighbors and hence it is trivially implemented in $O(1)$ rounds. Let $V^*$ be the remaining uncolored vertices. By the description above, the vertices compute the hierarchical decomposition into layers, blocks, stratum and super-blocks. In addition, all vertices know the virtual tree $\mathcal{T}$ which defines the relation between the blocks. We have that $V^*=(W_1^S, \ldots, W_s^S,W_1^L,\ldots, W_s^L,V_{sp})$. The vertices are colored in $s+2$ stages. First, all the small blocks are colored in $s$ phases: stratum by stratum. In other words, all the vertices in the small block are colored according in order: $W_s^S,\ldots, W_1^S$.
Next, the algorithm colors the vertices of $W'=\bigcup_{j=2}^s W^L_j$, i.e., all the vertices in large blocks, except for those belonging to blocks of the first layer $V_1$. Lastly, the vertices of the large block $W_1^S$ are colored. At the end of this hierarchy-based coloring, there is a (small) subset of vertices that failed to be colored and a subset of sparse vertices in $V_{sp}$). These vertices are colored by a similar procedure, a final cleanup phase. In the high-level, the CLP algorithm consists of two main procedures: a coloring procedure for the dense vertices and a coloring procedure of the sparse vertices. 

\paragraph{Coloring $\epsilon$-Dense Vertices.}
In the general setting, one is given $g$ subgraphs $S_1,\ldots, S_g$ (e.g., collection of $\epsilon$-almost cliques), each with weak diameter $2$. The coloring algorithm attempts to legally color a large fraction (as a function of the density) of the vertices in these subgraphs.
In Alg. $\DenseColoringStep$ of CLP, all vertices agree on a value $Z_{ex}$ which is a lower bound on the number of access colors w.r.t $S_j$. In addition, each $v \in S_j$ is associated with a parameter $D_v$. The important parameter is $\delta_v=D_v /Z_{ex}$. The algorithm will attempt to color $1-\delta_v$ fraction of the vertices. To color $S_j$ the verices are permuted in increasing order by the $D$-value, and are colored one by one based on the order given by the permutation. In the local model, since each subgraph $S_j$ has weak diameter $2$, this coloring step can be done in $O(1)$ rounds.

In the congested clique model, we cannot enjoy the small diameter of each clique. Instead we  
use the fact that each subgraphs $S_1,\ldots, S_g$ is \emph{small}. Since each $S_j$ is a \emph{super-block} it is a subset of some $\epsilon_i$-almost clique $C$ and each such clique has $(1+\epsilon)\Delta=O(\Delta)$ vertices.  By letting each vertex $v \in S_j$ send its palette of colors and its neighbors to the leader $v_j$ of $S_j$, the leader can locally simulate this permutation-based coloring. Overall, each leader needs to receive $O(\Delta^2)=O(n)$ messages, which can be done in $O(1)$ rounds for all the subgraphs $S_j$ simultaneously, using Lenzen's routing algorithm.
%
%it is sufficient that leader $v_j$ of each super-block $R_j$ will collect $O(\epsilon \Delta \cdot \Delta)$ messages regarding . That is, each vertex $u \in R_j$ needs to send $C$ the $O(\epsilon \Delta)$ vertices in $R_j$ which are \emph{not} neighbors of $u$, and its palette $O(\Delta)$ colors. Since $\epsilon<1$ and $\Delta=O(\sqrt{n})$, all leaders can receive this information in $O(1)$ rounds, locally simulate $\DenseColoringStep$ for their component and notify the color to each of the vertices. All vertices exchange their colors and if its legal, this is their final color.  

The vertices in the small blocks in each strata $W_k$ for $k\geq 2$ are colored by $O(1)$ applications of Alg. $\DenseColoringStep$. At that point, each vertex $v \in V^*$ will have at most $O(\epsilon^5 \Delta)$ uncolored neighbors in each layer $V_i$, with probability at least $1-exp(-\Omega(\poly\Delta))$. 
When $\Delta=\Omega(\log^4 n)$, after applying $\DenseColoringStep$ for $O(1)$ many times, the remaining uncolored part is partitioned into two subgraphs: the first has maximum degree $O(1)$ and hence can be locally colored in $O(1)$ rounds, by collecting it to one leader using Lenzen's routing. The second subgraph has, w.h.p., the property that each of its connected component has size $\poly\log n$. Then, by computing the connected components and collecting the topology of each connected component to the component' leader, this remaining graph can be colored in $O(1)$ many rounds.

Coloring the vertices in large blocks (other than $1$) is done applying $6$ times Alg. $\DenseColoringStep$. Coloring of the vertices in $V_1$ that belong to large blocks is more involved but with respect to its implementation in the congested clique, it uses the same tools as above. I.e., after applying $\DenseColoringStep$ for several times, the remaining subgraph can be colored locally in $O(1)$ many rounds (due to small components or constant maximal uncolored degree).
%
%%\paragraph{Coloring Vertices by Stratum: Small Blocks of $1$}
%\paragraph{Coloring Vertices by Stratum: Large Blocks other than $1$}
%%\paragraph{Coloring the Large Dense Regions.}
%
%\paragraph{Coloring Vertices by Stratum: Large Block of $1$}

\paragraph{Coloring Sparse Vertices and Vertices with Many Access Colors.}
The main procedure colors most of the vertices in $W_1,\ldots, W_s$ and leaves only a small fraction of uncolored vertices $U$, each has a large number of excess colors. The vertices in $U \cup V_{sp}$ are colored as follows. Consider first set of vertices in $U$ and orient that edges of $G[U]$ from the sparser endpoint to denser endpoint. Since each vertex knows the later of its neighbors, this is done in one round. If both endpoints are in the same layer, the edge is oriented towards the endpoint of smaller ID. 
The coloring of the previous steps guarantees that the outgoing degree of each vertex $v \in V_i$ is at most $O(\epsilon_{i-1}^{2.5}\cdot \Delta)$. By the initial coloring step, a vertex has $\Omega(\epsilon_{i-1}^2 \Delta)$ excess colors in its palette. The set of vertices $U$ are then colored using Alg. $\ColorBidding$ with runs in $O(\log^* \Delta)$ rounds in the local model (see Section 3.3 in \cite{CHP18}). It is easy to see that each vertex sends message of size $O(\log n/\epsilon)$ and hence it can be easily implemented in the congested clique model within the same number of rounds when $\Delta=O(\sqrt{n})$. The set of sparse vertices $V_{sp}$ is colored in a similar manner. All vertices in $U \cup V_{sp}$ that failed to be colored are added to the set of bad vertices $V_{bad}$. Overall, coloring the dense regions is done in $O(\log^* \Delta)$ rounds, and the clean-up phase of coloring the sparse region and the remaining uncolored vertices from strata at least $2$ is done in $O(\log^* \Delta)$ as well. Coloring the remaining uncolored vertices from strata $1$ is done in $O(1)$ rounds (this is the only step that requires $O(\log^* n)$ rounds, in the CLP algorithm). 

We next turn to prove \Cref{cor:colorsqrtnparallel} by showing that given a collection of vertex disjoint subgraphs with maximum degree $O(\sqrt{n})$, all subgraphs can be colored simultaneously in $O(\log^* \Delta)$ rounds. There are three main subroutines. Coloring the dense regions requires collecting the $2$-neighborhood topology, which can be done in $O(1)$ rounds, so long that $\Delta^2=O(n)$. Coloring sparse regions and the remaining uncolored vertices from strata at least $2$ also requires to have small maximum degree. Finally, consider the coloring of the remaining uncolored vertices from strata $1$. There are two types of left-over subgraphs. One with constant maximum degree -- all these subgraphs can be collected to a single leader in $O(1)$ rounds. The other subgraph contains components of poly-logarithmic size, since in such a case we color the subgraph by collecting information to the local leader of each component, this step can be done for all the 
vertex disjoint subgraphs simultaneously.

}%\APPENDFULLCLPSMALL

\paragraph{Handling Non-Equal Palette Size for $\Delta=O(\sqrt{n})$.} %\label{sec:nonequalclp}
The CLP algorithm assumes that each vertex is given a list of \emph{exactly} $(\Delta+1)$ colors. 
Our coloring algorithms requires a more relaxed setting where each vertex $v$ is allowed to be given a list of $r_v \in [\Delta-\Delta^{3/5},\Delta+1]$ colors where $r_v \geq \deg(v,G)+1$. 
In this subsection we show:
\vspace{-5pt}
\begin{lemma}\label{lem:modclp}
Given a graph $G$ with $\Delta=O(\sqrt{n})$, if every vertex $v$ has a palette with $r_v\geq \deg(v,G)+1$ colors and $r_v \in [\Delta-\Delta^{3/5},\Delta+1]$ then a list coloring can be computed in $O(\log^* \Delta)$ rounds in the congested clique model.
\end{lemma}
The key modification for handling non-equal palette sizes  is in 
definition of $\epsilon$-friend (which affects the entire decomposition of the graph). Throughout, let $q=\Delta^{3/5}$ and say\footnote{The value of $q$ is chosen to be a bit above the standard deviation of $\sqrt{\Delta\log n}$ that will occur in our algorithm.} that $u,v$ are $(\epsilon,q)$-friends if $|N(u)\cap N(v)|\geq (1-\epsilon)\cdot (\Delta-q)$. Clearly, if $u,v$ are $\epsilon$-friends, they are also $(\epsilon,q)$-friends. A vertex $v$ is an $(\epsilon,q)$-\emph{dense} if it has at least $(1-\epsilon)\cdot (\Delta-q)$ neighbors which are $(\epsilon,q)$-friends. An $(\epsilon,q)$-almost clique is a connected component of the subgraph induced by $(\epsilon,q)$-dense vertices and their $(\epsilon,q)$ friends edges. 
We next observe that for the $\epsilon$-values used in the CLP algorithm, the converse is also true up to some constant. 
\begin{observation}\label{obs:nonep}
For any $\epsilon \in [\Delta^{-10},K^{-1}]$, where $K$ is a large constant, and for $q=\Delta^{3/5}$, it holds that if $u,v$ are $(\epsilon,q)$ friends, they are $(2\epsilon)$-friends. Also, if $v$ is an $(\epsilon,q)$-dense, then it is $2\epsilon$-dense. 
\end{observation}
\begin{proof}
For any $\epsilon_1\leq \epsilon_2$, it holds that $(1-\epsilon_1)/\epsilon_1 \geq (1-\epsilon_2)/\epsilon_2$. 
Hence, for any $\epsilon \in [\Delta^{-10},K^{-1}]$, we have:
$(1-\epsilon)/\epsilon \leq \Delta^{1/10}\cdot (1-\Delta^{-1/10})\leq \Delta^{2/5},$
yielding that $(1-\epsilon)q \leq \epsilon \cdot \Delta$ for any $\epsilon \in [\Delta^{-10},K^{-1}]$. 
Since $u,v$ are $(\epsilon,q)$ friends, we have 
$|N(u)\cap N(v)|\geq (1-\epsilon)\cdot \Delta-(1-\epsilon)\cdot q \geq (1-2\epsilon) \cdot \Delta.$
If $v$ is an $(\epsilon,q)$-dense, then it has at least $(1-\epsilon)(\Delta-q)\geq (1-2\epsilon)\Delta$ neighbors which are $2\epsilon$-friends. 
\end{proof}
In \Cref{sec:modclp}, we provide a detailed proof for \Cref{lem:modclp}.
%In \Cref{sec:modclp}, we provide a more detailed proof.
\def\APPENDUNBAL{ 
\paragraph{Hierarchy, Blocks and Starta.}
The entire hierarchy of levels $V_1,\ldots, V_{\ell}$ is based on using the definition of $(\epsilon,q)$-dense vertices (rater than $\epsilon$-dense vertices).
Let $\bar{d}_{S,V'}(v)=|(N(v) \cap V')\setminus S|$ be the external degree of $v$ with respect to $S,V'$. Let $a_S(v)=|S \setminus (N(v)\cup \{v\})|$ be the anti-degree of $v$ with respect to $S$. Let $V_{\epsilon,q}^d, V_{\epsilon,q}^s$ be the vertices which are $(\epsilon,q)$-dense (resp., sparse).
By \Cref{obs:nonep}, for every $\epsilon \in [\Delta^{-10},1/K]$, it also holds:
\begin{observation}
For any $(\epsilon,q)$-almost clique $C$, there exists an $\epsilon$-almost clique $C_{\epsilon}$ and an $(2\epsilon)$-almost clique $C_{2\epsilon}$ such that $C_{\epsilon}\subseteq C \subseteq C_{2\epsilon}$. 
\end{observation}
As a result, Lemma 1 of \cite{CHP18} that describes the basic properties of cliques, holds up to small changes in the constants.
\begin{lemma}\label{lem:epsqclique}[Adopted from Lemma 1 \cite{CHP18}]
The following holds for an $(\epsilon,q)$-almost clique $C$, $\epsilon \leq 1/10$: (i) the external degree of $v \in C$ w.r.t. to $C$ is $\bar{d}_{C,V^d_{\epsilon,q}}(v)\leq 2\epsilon\Delta$, (ii) $a_C(v)\leq 6\epsilon \Delta$, (iii) $|C|\leq (1+6\epsilon)\cdot \Delta$, and (iv) $\dist(u,v,G)\leq 2$ for each $(u,v) \in C$, i.e., $C$ has \emph{weak} diameter $2$. 
\end{lemma} 
The notion of strata, blocks and super-blocks are all trivially extended using the definition of $(\epsilon,q)$-dense vertices. 
Note that since each vertex has a palette of at least $\deg+1$ colors, the excess of colors is non-decreasing throughout the coloring algorithm.

\paragraph{Initial Coloring Step.} In this step, Alg. $\OneShotColoring$ is applied for $O(1)$ many iterations. We change the coloring probability of Alg. $\OneShotColoring$ to be $p \in (0,1/8)$ rather than $p \in (0,1/4)$. It is then required to prove Lemma 5 of \cite{CHP18} and show that after applying $O(1)$ rounds of  Alg. $\OneShotColoring$, for every $(\epsilon,q)$-sparse vertex with $\deg(v)\geq 0.9\Delta$ it holds that: 
(i) it has at least $\Delta/2$ uncolored neighbors with high probability, and (ii) it has $\Omega(\epsilon^2 \cdot \Delta)$ excess colors with probability $1-exp^{-\Omega(\epsilon^2 \Delta)}$.  The proof is exactly the same up to small changes in the constants:
In Lemma 11 of \cite{CHP18}, we then get that $Pr[E_c]\leq p|S|/(\Delta-q)\leq 4p\leq 1/2$ as $p \leq 1/8$
In Lemma 12 of \cite{CHP18}, the condition on $Q$ is such that each $u_i \in S$ satisfies that $|\Phi(u_i)\cap Q|\geq (1-\epsilon/2)\cdot (\Delta-q)$. This only change a constant in the proof where $Pr[X_i=1]\geq p/4$. In Lemma 14, since $v$ is an $(\epsilon,q)$-sparse vertex, it is also an $\epsilon$-sparse vertex and thus the lemma again holds up to small changes in the constants. For each $(\epsilon,q)$ non-friend $u_i \in S$, we have $|S \setminus (S' \cup N(u_i))|\geq (\Delta-q)(1-\epsilon/5-\epsilon/100(1-\epsilon))\geq \epsilon \cdot \Delta/4$. All this effects again, only constants in the probabilities.

\paragraph{Main Coloring Step.}
Consider the coloring of small blocks which is based on Lemma 2 and Lemma 3 of \cite{CHP18}. Both these lemmas hold up to small changes in the constants. In particular, in Lemma 3, we are now given an $(\epsilon_i,q)$ almost clique $C$ and $C_1,\ldots, C_l$ are the $(\epsilon_{i-1},q)$-almost cliques contained in $C$. Since an $(\epsilon,q)$-almost clique is contained in an $2\epsilon$-almost clique, by using \Cref{lem:epsqclique}, we have that either $l=1$ or that $\sum_{j=1}^l |C_j|\leq 2(6\epsilon_i+2\epsilon_{i-1})\leq 14 \epsilon_i \cdot \Delta$. In the proof of Lemma 2, to bound $A_2$, using the bound of the external degree in \Cref{lem:epsqclique}, we get that $A_2\leq 4 \epsilon_{\ell}\Delta$.

Coloring the large blocks and the sparse vertices follows the exact same analysis an in \cite{CHP18}. 
The reason is that the coloring of the large blocks uses only the bounds on the external degrees and clique size (as shown in \Cref{lem:epsqclique}) and those properties hold up to insignificant changes in the constants (the algorithms of the sparse and dense components use the $O$-notation on these bounds, in any case). \Cref{lem:modclp} follows by combining the above with the implementation of the CLP algorithm in the congested clique model of \Cref{thm:colorsqrtn}.
}%\APPENDUNBAL{ 

\section{$(\Delta+1)$-Coloring for $\Delta=\Omega(\sqrt{n})$}
In this section, we describe a new recursive degree-reduction technique. As a warm-up, we start with  $\Delta=O((n/\log n)^{3/4})$. We make use of the following fact.
\begin{theorem} (\textbf{Simple Corollary of Chernoff Bound})\label{thm:cher} Suppose $X_1$, $X_2$, \dots, $X_\ell \in [0,1]$ are independent random variables, and let $X=\sum_{i=1}^{\ell} X_i$ and $\mu = \mathbb{E}[X]$. If $\mu \geq 5 \log n$, then w.h.p. $X \in \mu \pm \sqrt{5\mu\log n}$, and if $\mu < 5 \log n$, then w.h.p. $X \leq \mu +5\log n$.
\end{theorem}

%\vspace{-5pt}
\subsection{An $O(\log^* \Delta)$-round algorithm for $\Delta=O((n/\log n)^{3/4})$}\label{sec:deltathreefour}
%\vspace{-5pt}
The algorithm partitions $G$ into $O(\Delta^{1/3})$ subgraphs as follows. 
Let $\ell=\lceil \Delta^{1/3} \rceil$. We define $\ell+1$ subsets of vertices $V_1,\ldots, V_\ell, V^*$. 
A vertex joins each $V_i$ with probability 
$$p_i=1/\ell-2\sqrt{5\log n}/(\Delta^{1/3}\cdot \ell),$$
for every $i \in \{1,\ldots, \ell\}$, and it joins $V^*$ with the remaining probability of $p^*=2\sqrt{5\log n}/\Delta^{1/3}$. 

Let $G_i=G[V_i]$ be the induced subgraph for every $i \in \{1,\ldots, \ell,*\}$. 
Using Chernoff bound of \Cref{thm:cher}, the maximum degree $\Delta'$ in each subgraph $G_i$, $i \in \{1,\ldots,\ell\}$ is w.h.p.: $$\Delta'\leq \Delta/\ell-2\Delta^{2/3}\sqrt{5\log n}/\ell+\sqrt{5\Delta \log n/\ell}\leq  \Delta/\ell-1$$

In the first phase, all subgraphs $G_1, \ldots, G_{\ell}$ are colored independently and simultaneously.
This is done by allocating a distinct set of $(\Delta'+1)$ colors for each of these subgraphs. Overall, we allocate $\ell\cdot (\Delta'+1)\leq \Delta$ 
colors.  Since $\Delta'=O(\Delta^{2/3})=O(\sqrt{n})$, we can apply the $(\Delta'+1)$-coloring algorithm of \Cref{cor:colorsqrtnparallel} on all the graphs $G_1,\ldots, G_{\ell}$ simultaneously. Hence, all the subgraphs $G_1, \ldots, G_{\ell}$ are colored in $O(\log^* \Delta)$ rounds. 

\paragraph{Coloring the remaining left-over subgraph $G^*$.}
The second phase of the algorithm completes the coloring for the graph $G^*$. This coloring should agree with the colors computed for $G\setminus G^*$ computed in the previous phase. Hence, we need to color $G^*$ using a \emph{list} coloring algorithm. 
We first show that w.h.p. the maximum degree $\Delta^*$ in $G^*$ is $O(\sqrt{n})$. 
The probability of vertex to be in $G^*$ is $p^*=2\sqrt{5\log n}/\Delta^{1/3}$. By Chernoff bound of \ref{thm:cher}, w.h.p., $\Delta^*\leq p^* \cdot \Delta+\sqrt{5p^* \cdot \Delta \cdot \log n}$. Since $\Delta\leq (n/\log n)^{3/4}$, $\Delta^*=O(\sqrt{n})$.  To be able to apply the modified CLP of \Cref{lem:modclp}, we show:
\begin{lemma}\label{lem:size}
Every $v \in G^*$ has at least $\Delta^*-(\Delta^*)^{3/5}$ available colors in its palette after coloring all its neighbors in $G\setminus G^*$. 
\end{lemma} 
\begin{proof}
First, consider the case where $\deg(v,G)\leq \Delta-(\Delta^*-\sqrt{5\Delta^* \cdot \log n})$. In such case, even after coloring all neighbors of $v$, it still has an access of $\Delta^*-\sqrt{5\Delta^* \cdot \log n}\geq \Delta^*-(\Delta^*)^{3/5}$ colors in its palette after coloring $G\setminus G^*$ in the first phase. 
Now, consider a vertex $v$ with $\deg(v,G)\geq \Delta-(\Delta^*-\sqrt{5\Delta^* \cdot \log n})$. 
Using Chernoff bound, w.h.p., 
$\deg(v,G^*)>(\Delta-(\Delta^*-\sqrt{\Delta^* \cdot 5\log n}))\cdot p^*-\sqrt{5\log n \Delta p^*}\geq \Delta^*-(\Delta^*)^{3/5}.$
\end{proof}
Also note that a vertex $v \in G^*$ has at least $\deg(v,G^*)+1$ available colors, since all its neighbors in $G^*$ are uncolored at the beginning of the second phase and initially it was given $(\Delta+1)$ colors.
Eventhough, $v \in G^*$ might have $\Omega(\Delta)$ neighbors not in $G^*$, to complete the coloring of $G^*$, by \Cref{lem:size}, after the first phase, each $v$ can find in its palette $r \in [\Delta^*-(\Delta^*)^{3/5},\Delta^*+1]$ available colors and this sub-palette is sufficient for its coloring in $G^*$. Since $\Delta^*=O(\sqrt{n})$, to color $G^*$ (using these small palettes), one can apply the $O(\log^* \Delta)$ round list-coloring algorithm of \Cref{lem:modclp}.
%
%
 %that agree with the coloring for $G\setminus G^*$, it is sufficient for each vertex in $G^*$ to keep only a palette with $\Delta^*+1$ colors among all the available (possibly $\Omega(\Delta)$) colors that it currently has. 
%That is, each vertex $v \in G^*$, should choose only $\Delta^*+1$ colors from its palette, eventhough its original palette (after the coloring of the first phase) might have been much larger. Once we have that, we can apply a $(\Delta^*+1)$-list coloring algorithm of \Cref{thm:colorsqrtn}.
\vspace{-5pt}
\subsection{An $O(\log(1/\epsilon)\cdot\log^*\Delta)$-round algorithm for $\Delta=O((n/\log n)^{1-\epsilon})$}\label{sec:deltaeps}
Let $N=n/(5\log n)$.  First assume that $\Delta\leq N/2$ and partitions the range of relevant degrees $[\sqrt{n},N/2]$ into $\ell=\Theta(\log\log \Delta)$ classes.
The $y^{th}$ range contains all degrees in $[N^{1-1/2^y},N^{1-1/(2^{y+1})}]$ for every $y \in \{1,\ldots, \ell\}$. Given a graph $G$ with maximum degree $\Delta =O(N^{1-1/(2^{y+1})})$, Algorithm $\RecursiveColor$ colors $G$ in $y \cdot O(\log^* \Delta)$ rounds, w.h.p.

\paragraph{Step (I): Partitioning (Defective-Coloring).}
For $i\in \{0,\ldots, y-1\}$, in every level $i$ of the recursion, we are given a graph $G'$, with maximum degree $\Delta_i=O(N^{1-1/2^{y-i+1}})$, and a palette $\Pal_i$ of $(\Delta_i+1)$ colors. 
For $i=0$, $\Delta_0=\Delta$ and the palette $\Pal_0=\{1,\ldots, \Delta+1\}$.

The algorithm partitions the vertices of $G'$ into $q_i+1$ subsets: $V'_1,\ldots, V'_{q_i}$ and a special left-over set $V^*$. The partitioning is based on the following parameters. 
Set $x=2^{y-i}$ and 
$$q_i=\lceil \Delta_i^{1/(2x-1)}\rceil \mbox{~~and~~} \delta_i=2\sqrt{5\log n}\cdot q_i^{3/2}/\sqrt{\Delta_i}.$$
%Since $x\geq 2$, it holds that $\delta_i \in (0,1)$. 

Each vertex $v \in V(G')$ joins $V'_j$ with probability $p_j=1/q_i-\delta_i/(q_i)^2$ for every  $j \in \{1,\ldots, q_i\}$, and it joins $V^*$ with probability $p^*=\delta_i/q_i$. Note that $p^* \in (0,1)$ as $x\geq 2$. For every $j \in \{1,\ldots, q_i\}$, let $G'_{j}=G'[V'_j]$ and let $G^*=G'[V^*]$. 

\paragraph{Step (II): Recursive coloring of $G'_1,\ldots, G'_{q_i}$.}
Denote by $\widetilde{\Delta}_j$ to be the maximum degree in $G'_j$ for every $j \in \{1,\ldots, q_i\}$ and by $\Delta^*$, the maximum degree in $G^*$.  The algorithm allocates a distinct subset of $(\widetilde{\Delta}_j+1)$ colors from $\Pal_i$ for every $j \in \{1,\ldots, q_i\}$. In the analysis, we show that w.h.p. $\Pal_i$ contains sufficiently many colors for that allocation. The subgraphs $G'_1,\ldots, G'_{q_i}$ are colored recursively and simultaneously, each using its own palette. 
It is easy to see that the maximum degree of each $G'_j$ is $O(N^{1-1/2^{y-i}})$ (which is indeed the desire degree for the subgraphs colored in level $i+1$ of the recursion).

\paragraph{Step (III): Coloring the left-over graph $G^*$.}
Since the algorithm already allocated at most $\Delta_i$ colors for coloring the $G'_j$ subgraphs, it might run out of colors to allocate for $G^*$. This last subgraph is colored using a list-coloring algorithm  only after all vertices of $G'_1,\ldots, G'_{q_i}$ are colored. Recall that $\Delta^*$ is the maximum degree of $G^*$. In the analysis, we show that w.h.p. $\Delta^*=O(\sqrt{n})$. For every $v \in G^*$, let $\Pal(v) \subseteq \Pal_i$ be the remaining set of available colors after coloring all the vertices in $V'_1,\ldots, V'_{q_i}$. 
Each vertex $v \in G^*$ computes a new palette $\Pal^*(v) \subseteq \Pal(v)$ such that: (i) $|\Pal^*(v)|\geq \deg(v, G^*)+1$, and (ii) $|\Pal^*(v)| \in [\Delta^*-\Delta^{3/5},\Delta^*+1]$. 
In the analysis section, we show that w.h.p. this is indeed possible for every $v \in G^*$. 
The algorithm then applies the modified CLP algorithm, and $G^*$ gets colored within $O(\log^* \Delta)$ rounds.

\setlength{\columnsep}{19pt}%
\begin{wrapfigure}{r}{7cm}
\centering
\vspace{-5pt}
\includegraphics[width=0.30\textwidth]{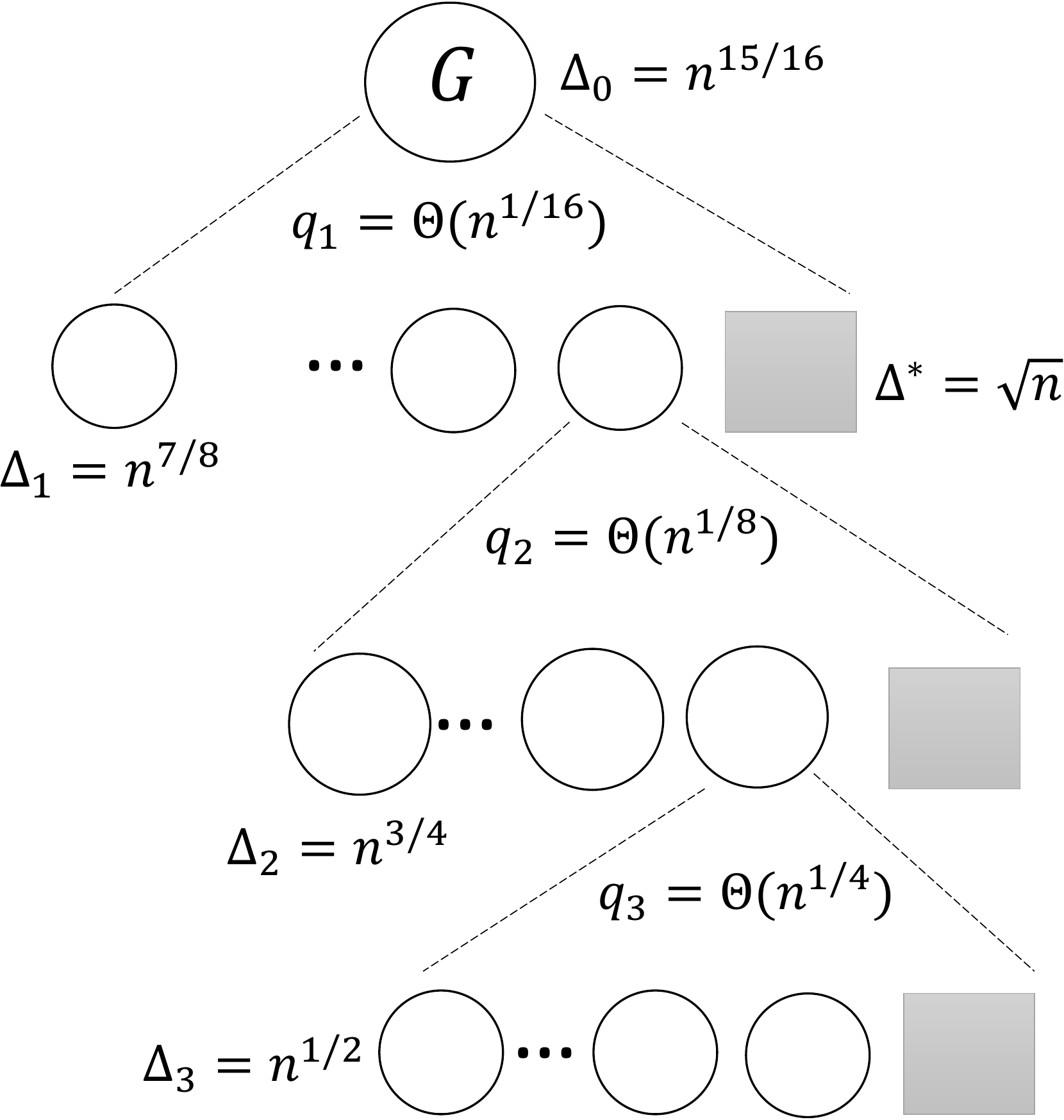}
\vspace{-5pt}
	\caption{
		\vspace{-15pt}}
	%\label{fig:recursive}
\end{wrapfigure} 
\paragraph{Example:}
Assume that input graph $G$ has maximum degree $\Delta_0=n^{15/16}$. The algorithm partitions $G$ into $k_0=n^{1/15}$ subgraphs in the following manner. For sake of clarity, we omit logarithmic factors in the explanation. With probability $1/(\Delta_0)^{1/2+o(1)}$, $v$ joins a left-over subgraph $G^*$, and with the remaining probability it picks a subgraph $[1,k_0]$ uniformly at random. It is easy to see that the maximum degree in each of these subgraphs is at most $\Delta_1=n^{7/8}+n^{7/16}$. A distinct set of $\Delta_1$ colors from $[1,\Delta_0+1]$ is allocated to each of the $k$ subgraphs. Each such subgraph is now partitioned into $k_1=n^{1/8}$ subgraphs plus a left-over subgraph. This continues until all subgraphs have their degrees sharply concentrated around $\sqrt{n}$. At that point, the modified CLP algorithm can be applied on all the subgraphs in the last level $\ell$. Once these subgraphs are colored, the left-over subgraphs in level $\ell-1$ are colored, this continues until the final left-over subgraph of the first level is colored.
We next provide a compact and high level description of the algorithm.
\begin{mdframed}[hidealllines=false,backgroundcolor=gray!30]
\center \textbf{Algorithm $\RecursiveColor(G',\Pal_i)$}
\begin{flushleft}
Input: Graph $G'$ with maximum degree $\Delta_i=O(N^{1-1/2^{y-i+1}})$. \\
A palette $\Pal_i$ of $(\Delta_i+1)$ colors (same for all nodes). 
\end{flushleft}
\vspace{-10pt}
\begin{itemize}
\item
Partitions $G'$ into $q_i+1$ vertex-disjoint subgraphs:
\begin{itemize}
\item
$q_i$ vertex-subgraphs $G'_1,\ldots, G'_{q_i}$ with maximum degree $\Delta_{i+1}=O(N^{1-1/2^{y-i}})$.
\item
Left-over subgraph $G^*$ with maximum degree $\Delta^*=O(\sqrt{n})$.
\end{itemize}
\item
Allocate a distinct palette $\Pal_j \subset \Pal_i$ of $(\Delta_{i+1}+1)$ colors for each $j \leq q_i$.
\item
Apply $\RecursiveColor(G_j,\Pal_j)$ for every $j\leq q_i$ simultaneously.
\item
Apply a $(\Delta_i+1)$-list coloring restricted to $\Pal_i$, to complete the coloring of $G[V^*]$.
\end{itemize}
\end{mdframed}
\paragraph{Analysis.}

\begin{lemma}
(i) For every $j \in \{1,\ldots, q_i\}$, w.h.p., $\widetilde{\Delta}_j=O(N^{1-1/2^{y-i}})$. 
(ii) One  can allocate $(\widetilde{\Delta}_j+1)$ distinct colors from $\Pal_i$ for each $G'_j$,  $j \in \{1,\ldots, q_i\}$.
\end{lemma}
\begin{proof}
Using Chernoff bound of \Cref{thm:cher}, w.h.p., for every $j \in \{1,\ldots, q_i-1\}$, the maximum degree $\widetilde{\Delta}_j$ in $G'_j$ is at most
$\widetilde{\Delta}_j=O(\Delta_i/q_i)$. 
Since $\Delta_i=O(N^{1-2^{y-i+1}})$, claim (i) follows.
We now bound the sum of all colors allocated to these subgraphs:
$$\widetilde{\Delta}_j\leq \Delta_i/q_i-(\Delta_i\cdot \delta_i)/(q_i)^2+\sqrt{5\log n \cdot \Delta_i/q_i}\leq \Delta_i/q_i-1~.$$
where the last inequality follows by the value of $\delta_i$. 
We get that $\sum_{j=1}^{q_i}(\widetilde{\Delta}_j+1)\leq \Delta_i$ and since $\Pal_i$ contains $\Delta_i+1$ colors, claim (ii) follows.
\end{proof}
We next analyze the final step of the algorithm and begin by showing that, w.h.p., the maximum degree in the left-over graph $G^*$ is $O(\sqrt{n})$.  By Chernoff bound of \Cref{thm:cher}, w.h.p., the maximum degree $\Delta^*\leq \Delta_i \cdot  \delta_i/q_i+\sqrt{\log n\cdot \Delta_i \cdot  \delta_i/q_i}$.  Since $\Delta_i=O(N^{1-1/2^{y-i+1}})$, we get that $\Delta^*=O(\sqrt{n})$. 
We now claim:
\begin{lemma}
After coloring for all the vertices in $G'_1,\ldots, G'_{q_i}$, each vertex $v \in G^*$ has a palette $\Pal^*(v)$ of free colors such that (i) $|\Pal^*(v)|\geq \deg(u,G^*)+1$, and (ii) $|\Pal^*(v)| \in [\Delta^*-(\Delta^*)^{3/5}),\Delta^*+1]$. 
\end{lemma}
\begin{proof}
Since each vertex $v \in G'$ has a palette of size $(\Delta_i+1)\geq \deg(v,G')$, after coloring all its neighbors in $G'_1,\ldots, G'_{q_i}$, it has at least $\deg(v,G^*)+1$ free colors in its palette.  Claim (ii) follows the same argument as in \Cref{lem:size}. 
We show that the palette of $v$ has at least $\Delta^*-O(\sqrt{\Delta^*}\cdot 5\log n)\geq (\Delta^*-(\Delta^*)^{3/5})$ available colors after coloring all the vertices in $G\setminus G^*$. 
First, when $\deg(v,G')\leq \Delta_i-(\Delta^*-\sqrt{\Delta^* \cdot 5\log n})$, then even after coloring all neighbors of $v$ in $G'$, it still has an access of $\Delta^*-\sqrt{\Delta^* \cdot 5\log n}$ colors in its palette. Consider a vertex $v$ with $\deg(v,G')\geq \Delta_i-(\Delta^*-\sqrt{\Delta^* \cdot 5\log n})$. By Chernoff, w.h.p. it holds that:
\begin{eqnarray*}
\deg(v,G^*)&\geq &(\Delta_i-(\Delta^*-\sqrt{\Delta^* \cdot 5\log n}))\cdot p^*-\sqrt{5\log n\cdot \Delta_i \cdot p^*}
\\&\geq& \Delta_i\cdot \delta_i/q_i-O(\sqrt{\Delta_i\cdot \delta_i \log n/q_i})
\geq
\Delta^*-(\Delta^*)^{3/5}~.
\end{eqnarray*}
Hence, by combining with claim (i), the lemma follows.
\end{proof}
This completes the proof of \Cref{thm:coleps}(i). 

\paragraph{$(\Delta+\Delta^{1/2+\epsilon})$ Coloring in Log-Star Rounds}
\begin{lemma}\label{lem:manycol}
For any fixed $\epsilon \in (0,1)$, one can color, w.h.p., a graph with $(\Delta+\Delta^{1/2+\epsilon})$ colors in $O(\log(1/\epsilon) \cdot \log^* \Delta)$ rounds.
\end{lemma}
\begin{proof}
Due to \Cref{thm:colorsqrtn}, it is sufficient to consider the case where $\Delta=\Omega(\sqrt{n})$. 
Partition the graph into $k=\lfloor \Delta^{\epsilon} \rfloor$ subgraphs $G_1,\ldots, G_k$, by letting each vertex independently pick a subgraph uniformly at random. By Chernoff bound of \Cref{thm:cher}, the maximum degree $\Delta_i$ in each subgraph $G_i$ is at most $\Delta_i\leq \Delta^{1-\epsilon}+\Delta^{1/2-\epsilon/2}\cdot \sqrt{5\log n}$. Allocate a distinct set of $\Delta_i+1$ colors $\Pal_i$ to each subgraph $G_i$. Since $\Delta^{1-\epsilon}=O((n/\log n)^{1-\epsilon/2})$, we can apply 
Alg. $\RecursiveColor$ on each of these subgraphs which takes $O(\log(1/\epsilon)\cdot \log^* \Delta)$ rounds. It is easy to see, that since the subgraphs are vertex disjoint, Alg. $\RecursiveColor$ can be applied on all $k$ subgraphs simultaneously with the same round complexity. Overall, the algorithm uses 
$\Delta+\Delta^{1/2+\epsilon}$ colors.
\end{proof}

\vspace{-13pt}
\subsection{$(\Delta+1)$ Coloring Algorithm for General Graphs}\label{sec:generalcol}
For graphs $G$ with $\Delta\leq n/(10\log n)$, we simply apply Alg. $\RecursiveColor$. Plugging $\epsilon=1/\log n$ in \Cref{thm:coleps}, we get that this is done in $O(\log\log \Delta \cdot \log^* \Delta)$ rounds. 
It remains to handle graphs with with $\Delta \in [n/(10\log n),n]$. 
We partition the graph into $\ell=\lceil 5 \log n \rceil$ subgraphs $G_1,G_2, \ldots, G_\ell$ and a left-over graph $G^*$ in the following manner.
Each $v \in V$ joins $G_i$ with probability $p=1/\ell-2\sqrt{5\log n/(\Delta \cdot \ell)}$ for every $i \in \{1,\ldots, \ell\}$, and it joins $G^*$ with probability $p^*=1-\ell \cdot p=\Theta(\log n/\sqrt{\Delta})$. By Chernoff bound, the maximum degree in $G_i$ for $i \in \{1,2,\ldots,\ell\}$ is 
$\Delta_i\leq \Delta/\ell-2\sqrt{(\Delta \cdot 5\log n)/\ell}+\sqrt{\Delta \cdot 5\log n/\ell}\leq \Delta/\ell-1~.$
Hence, we have the budget to allocate a distinct set $\Pal_i$ of $\Delta_i$ colors for each $G_i$.

The first phase applies Algorithm $\RecursiveColor$ on each  $(G_i, \Pal_i)$ simultaneously for every $i$. Since $\Delta_i=O(n/\log n)$, and the subgraphs are vertex-disjoint, this can be done in $O(\log\log \Delta \cdot \log^* \Delta)$ rounds for all subgraphs simultaneously (see \Cref{thm:coleps}(i)). 

After all the vertices of $G\setminus G^*$ get colored, the second phase colors the left-over subgraph $G^*$. The probability of a vertex $v$ to be in $G^*$ is $O(\log n /\sqrt{\Delta})=O(\log^2 n/\sqrt{n})$. Hence, $G^*$ contains $O(\log^2 n \cdot \sqrt{n})$ vertices with high probability. We color $G^*$ in two steps. First, we use the $\deg+1$ list coloring Algorithm $\OneShotColoring$ from \cite{barenboim2016locality} to reduce the uncolored-degree of each vertex to be $O(\sqrt{n}/\log^2 n)$ with high probability. This can be done in $O(\log\log n)$ rounds. In the second step, the entire uncolored subgraph $G'' \subset G^*$ has $O(n)$ edges and can be solved locally in $O(1)$ rounds. Note that for each $v \in G''$, it is sufficient to consider a palette with $\deg(v,G'')+1$ colors, and hence sending all these palettes can be done in $O(1)$ rounds as well. 

We are now ready to complete proof of \Cref{thm:maincol}. 
The correctness of the first phase follows by \Cref{thm:coleps}(i). Note that since the $G_i$'s subgraphs are vertex-disjoint, they can handled simultaneously by Alg. $\RecursiveColor$. Hence the first coloring phase takes $O(\log\log \Delta\cdot \log^* \Delta)$ rounds. Finally, for the second phase, we use Lemma 5.4 of \cite{barenboim2016locality}. Since we only want to reduce the uncolored degree of each vertex in $G^*$ to be $O(\sqrt{n}/\log^2 n)$, i.e., reduce it be a factor of at most $\log^4 n$, using Lemma 5.4, the uncolored degree of each (relevent) vertex $v$  is reduced by a constant factor w.h.p. and hence after $O(\log\log n)$ rounds, all degrees are $O(\sqrt{n}/\log^2 n)$ as desired, with high probability.
\begin{mdframed}[hidealllines=false,backgroundcolor=gray!30]
\center \textbf{Algorithm $\GeneralColor(G)$}
%\begin{flushleft}
%Input: Graph $G$ with maximum degree $\Delta_i=O(N^{1-1/2^{y-i+1}})$; \\A palette $\Pal_i$ of $(\Delta_i+1)$ colors (same for all node). 
%\end{flushleft}
%\vspace{-10pt}
\begin{itemize}
\item
If $\Delta \leq n/(10\log n)$, call $\RecursiveColor(G, [1,\Delta+1])$. 
\item
Else, partition $G$ into vertex-subgraphs as follows:
\begin{itemize}
\item
$G'_1,\ldots, G'_q$ with the maximum degree $\Theta(\Delta/\log n)$, and 
\item
a left-over subgraph $G^*$ with maximum degree $\Delta^*=O(\sqrt{n})$.
\end{itemize}
\item
Allocate a distinct palette $\Pal_j \subset [1,\Delta+1]$ of $(\Delta(G_{j})+1)$ colors for each $j \leq q_i$.
\item
Apply $\RecursiveColor(G_j,\Pal_j)$ for all $G_1,\ldots, G_q$ simultaneously.
\item
Apply a $(\deg+1)$-list coloring algorithm on $G^*$ for $O(\log\log n)$ rounds.
\item
Solve the remaining uncolored subgraph locally. 
\end{itemize}
\end{mdframed}

\section{Deterministic Coloring Algorithms}\label{sec:detcol}

\paragraph{$O(\Delta^2)$ Coloring in $O(1)$ Rounds.} As a warm-up, we start by showing a very simple algorithm for computing $\Delta^2$ coloring very fast. 
\begin{lemma}
There exists a $O(1)$-round \emph{deterministic} algorithm for $\Delta^2$ coloring. 
\end{lemma}
\begin{proof}
Consider the following single round randomized algorithm: each node picks a random color in $[1,\ldots,\Delta^2]$. Nodes exchange their colors with their neighbors and if their color is legal they halt. It is easy to see that a node remains uncolored with probability $1/\Delta$ and this holds even if the random choices are pairwise independent. Hence, if all nodes are given a random seed of length $O(\log n)$, in expectation the number of remaining uncolored vertices is $n/\Delta$. To derandomize this single randomized step, we will use the method of conditional expectation similarly to the algorithm of \Cref{thm:detcoldelta}. We split the seed into chunks of $\lfloor \log n\rfloor$ pieces and describe how to compute the $i^{th}$ chunk using $O(1)$ rounds, given that the first $i-1$ chunks are already computed. 
We assign a special vertex to each of the $n$ possible assignments of an $\lfloor \log n\rfloor$-size (binary) chunk. Each vertex $u$, computes the probability $p_u(b)$ that it is legally colored given that the assignment to the $i^{th}$ chunk is $b$. Note that $u$ can compute this probability as it depends only at its neighbors and by knowing the IDs of its neighbors, $u$ can simulate their choices (given a seed). Each node $u$ sends the outcome $p_u(b)$ to the node responsible for the assignment $b$. Finally, the assignment that got the largest fraction of removed nodes is elected. Thanks to the method of conditional expectation, when all nodes simulate their random decision using the computed seed, the fraction of removed nodes is at least as the expected one and hence at most $n/\Delta$ vertices remained uncolored. At the point, the remaining graph can be collection to a single node and be locally solved. 
\end{proof}

We next turn to consider the more challenging task of computing $(\Delta+1)$ coloring in $O(\log \Delta)$ rounds. We will first describe a simple algorithm for the case where $\Delta=O(\sqrt{n})$. In this regime, the degrees are small enough to allow each node learning the palettes of its neighbors. Then, we will modify this basic algorithm to allow fast computation even for $\Delta=O(n^{3/4})$. Finally, using the partitioning technique described in the first part of the paper, we will handle the general case. 
\begin{theorem}\label{thm:detcoldelta}
There is a deterministic $(\Delta+1)$ list coloring using $O(\log \Delta)$ rounds, in the congested clique model.
\end{theorem}
In \cite{Censor-HillelPS17}, a deterministic $(\Delta+1)$ coloring was presented only for graphs with maximum degree $\Delta=O(n^{1/3})$. Here, we handle the case of $\Delta=\Omega(n^{1/3})$. 

\subsection{$(\Delta+1)$ Coloring for $\Delta=O(\sqrt{n})$}\label{sec:deg-very-small}
We first let each node sends its palette to all its neighbors. Since $\Delta=O(\sqrt{n})$, this can be done in $O(1)$ rounds. We will derandomize the following simple $(\Delta+1)$-algorithm that runs in $O(\log n)$ rounds.
\begin{mdframed}[hidealllines=false,backgroundcolor=gray!25]
\textbf{Round $i$ of Algorithm $\SimpleRandColor$ (for node $v$ with palette $\Pal_v$)}
\begin{itemize}
\item
Let $\Pal_{i,v}$ be the current palette of $v$ containing all its colors that are not yet taken by its neighbors. Let $F_{i,v}=|\Pal_{i,v}|$.
\item
With probability $1/2$, let $c_v=0$ and w.p. $1/2$ let $c_v$ be chosen uniformly at random from $\Pal_{i,v}$.
\item
Send $c_v$ to all neighbors and if $c_v \neq 0$ and legal, halt. 
\end{itemize}
\end{mdframed}

\begin{observation}\label{obs:pwenough}
The correctness of Algorithm $\SimpleRandColor$ is preserved, even if the coin flips are \emph{pairwise}-independent. 
\end{observation} 
\begin{proof}
We analyze the probability that some vertex $v \in V$ terminates in round $i$, conditioned that it has not terminated before round $i$, for any $i>0$. The probability that $v$ picked in round $i$ a color $c_v\neq 0$ is $1/2$.
Suppose that $c_v \neq 0$ and consider some neighbor $u$ of $v$ that has not terminated in round $i$.
The probability that $c_v=c_u$ is $1/2(F_{i,v})$. This is because the probability that $c_u\neq 0$ is $1/2$ and the probability that $v$ picked that color among the $F_{i,v}$ is $1/F_{i,v}$. Note that since this argument is only for two neighbors $u,v$, pairwise independence is enough. Let $\deg_i(v)$ be the number of uncolored neighbors of $v$ at the beginning on round $i$. Clearly, $\deg_i(v) \leq F_{i,v}+1$. By applying the union bound over all non-coloring neighbors $u$ of $v$, we get that $v$ is colored with some color in $\Pal_{i,v}$ with probability $1/2$. Hence, overall $v$ is colored with probability $1/4$.
\end{proof}
%We now show how to implement round $i$ of Algorithm $\SimpleRandColor$ using $O(1)$ deterministic rounds. 
The goal of \emph{phase} $i$ in our algorithm is to compute a seed that would be used to simulate the random color choices of \emph{round} $i$ of Alg. $\SimpleRandColor$. This seed will be shown to be good enough so that at least $1/4$ of the currently uncolored vertices, get colored when picking their color using that seed. Let $V_i$ be the set of uncolored vertices at the beginning of phase $i$.
We need the following construction of bounded independent hash functions:
\begin{lemma}\label{lem:vad}\cite{Vadhan12}
	\label{lem: d-wise independent}
	For every $\gamma,\beta,d \in \mathbb{N},$ there is a family of $d$-wise independent functions $\mathcal{H}_{\gamma,\beta} = \set{h : \set{0,1}^\gamma \rightarrow \set{0,1}^\beta}$ such that choosing a random function from $\mathcal{H}_{\gamma,\beta}$ takes $d \cdot \max \set{\gamma,\beta}$ random bits, and evaluating
	a function from $\mathcal{H}_{\gamma,\beta}$ takes time $poly(\gamma,\beta,d)$.
\end{lemma}
First, at the beginning of phase $i$, we let each node $u$ send its current palette $\Pal_{i,u}$ to all its neighbors. Since $\Delta=O(\sqrt{n})$, this can be done in $O(1)$ rounds.
For our purposes, to derandomize a single round, we use \Cref{lem:vad} with $d=2, \gamma=\log n, \beta=\log \Delta$ and hence the size of the random seed is $\alpha \cdot \log n$ bits for some constant $\alpha$. Instead of revealing the seed bit by bit using the conditional expectation method, we reveal the assignment for a \emph{chunk} of $z=\lfloor \log n \rfloor$ variables at a time.
To do so, consider the $i$'th chunk of the seed $Y'_i=(y'_1,\ldots, y'_z)$. For each of the $n$ possible assignments $(b'_1,\ldots, b'_z) \in \{0,1\}^{z}$ to the $z$ variables in $Y'$, we assign a leader $u$ that represent that assignment and receives the conditional expectation values from all the uncolored nodes $V_i$, where the conditional expectation is computed based on assigning $y'_1=b'_1,\ldots, y'_z=b'_z$. Unlike the MIS problem, here the vertex's success depends only on its neighbors (i.e., and does not depend on its second neighborhood). Using the partial seed and the IDs of its neighbors, every vertex $v$ can compute the probability that it gets colored based on the partial seed, its own palette and the palettes of its neighbors. It then sends its probability of being colored using a particular assignment $y'_1=b'_1,\ldots, y'_z=b'_z$ to the leader $u$ responsible for that assignment. 
The leader node $u$ of each assignment $y'_1=b'_1,\ldots, y'_z=b'_z$ sums up all the values and obtains the expected number of colored nodes conditioned on the assignment. Finally, all nodes send to the leader their computed sum and the leader selects the assignment $(b^*_1,\ldots, b^*_z) \in \{0,1\}^{z}$ of largest value.
After $O(1)$ many rounds, the entire assignment of the $O(\log n)$ bits of the seed are revealed. Every yet uncolored vertex $v \in V_i$ uses this 
seed to simulate the random choice of Alg. $\SimpleRandColor$, that is selecting a color in $\{0,1,2,\ldots, \Delta+1\}\setminus F_v$ and broadcasts its decision to its neighbors. If the color $c_v \neq 0$ is legal, $v$ is finally colored and it notifies its neighbors. By the correctness of the conditional expectation approach, we have that least $1/4 \cdot |V_i|$ vertices got colored.  Hence, after $O(\log n)=O(\log \Delta)$ rounds, all vertices are colored. %We next also show a deterministic $O(\Delta^2)$ coloring in $O(1)$ rounds. 

\subsection{$(\Delta+1)$ Coloring for $\Delta\leq n^{3/4}$}\label{sec:threequarter}
In this section, we show that one can modify the algorithm of the previous section, so that it can be implemented efficiently already for larger values of $\Delta \in [\sqrt{n},n^{4/3}]$.  
\begin{lemma}\label{lem:det-small-deg}
There exists a deterministic algorithm that given a graph $G=(V,E)$ of maximum degree $\Delta=O(n^{3/4})$ where each node $u$ has a palette $\Pal(u)$ of at least $\deg(u,G)+1$ colors in the range $[1,\ldots, \Delta+1]$ computes a list-coloring in $O(\log \Delta)$ rounds. 
\end{lemma}
We will have $O(\log \Delta)=O(\log n)$ phases, each will color at least a constant faction of the nodes. Each phase $j$ will be implemented in two steps. The first step computes a large subset of nodes $A'$ that contains at least a constant fraction of the currently uncolored nodes, along with a set of colors $S(u)\subseteq \Pal(u)$ for each node $u \in A'$. These $S(u)$ sets will have two desired properties.  On the one hand, they will be large enough so that when running a single (slightly modified) step of Algorithm $\SimpleRandColor$ using the $S(u)$ sets as the palettes, each node will be colored with constant probability. On the other hand, they will be sufficiently small to allow each node to send it to all its \emph{relevant} neighbors. We will say that $u$ and $v$ are relevant neighbors if they are neighbors in $G$ and $S(u) \cap S(v) \neq \emptyset$. The number of relevant neighbors of a node serves as an upper bound on the number of its competitors on a given color.
The second step of the phase will then simulate a single round of Algorithm $\SimpleRandColor$ in a similar way to the algorithm of Sec. \ref{sec:deg-very-small} up to minor modifications. We first describe a randomized algorithm that uses only $O(1)$-wise independence and then explain how to derandomize it efficiently. 

\subsubsection{Randomized Algorithm with $O(1)$-Wise Independence (Phase $j$).} 
For each node $u$, let $\Pal(u)$ be the set of currently free colors in the palette of $u$ at the beginning of phase $j$, and denote by $F_u=|\Pal(u)|$. Let $V_j$ the set of uncolored nodes at the begining of the phase and $\Gamma_j(u)=\Gamma(u)\cap V_j$, i.e., the set of currently uncolored neighbors of $u$. Clearly, $F_u > |\Gamma_j(u)|$. 

\paragraph{Step 1.} 
The goal is to compute a subset $A'$ and a set of colors $S(v) \subseteq \Pal(v)$ for each $v \in A'$, that will be used as the palettes in Step 2. 
Every uncolored node $u$ partitions its palette $\Pal(u)$ into $\Delta^{1/3}$ bins, each corresponds to a consecutive range of $\Delta^{2/3}$ colors. Specifically, for every $i \in \{1,\ldots, \Delta^{1/3}-1\}$, let $B_{i,u}$ be the set of free colors in $\Pal(u)$ restricted to the range $[(i-1)\Delta^{1/3}+1, i \cdot\Delta^{1/3}]$. Formally, 
$$B_{i,u}=[(i-1)\Delta^{1/3}+1, i \cdot\Delta^{1/3}] \cap \Pal(u) \mbox{~~and~~} B_{\Delta^{1/3}}=[\Delta-\Delta^{1/3}+1,\Delta+1] \cap \Pal(u)~.$$ 

A bin $B_{i,u}$ is \emph{small} if $|B_{i,u}|\leq \Delta^{1/12}$ and otherwise it is \emph{large}. 
Let $A_0$ be the subset of all currently uncolored nodes that at least $1/10$ fraction of their free colors are in small bins. That is, $A_0=\{ u \in V_j ~\mid~ \sum_{i: |B_{i,u}|\leq \Delta^{1/12}}|B_{i,u}|\geq F_u/10\}$.
Let $A_1=V_j \setminus A_0$ be the set of remaining uncolored nodes. 
First assume that $|A_0|\geq |A_1|$. In this case, $A'=A_0$, and for every $u \in A'$, let $S(u)=\bigcup_{i: |B_{i,u}|\leq \Delta^{1/12}} B_{i,u}$ be the set of all the colors in the small bins of $u$.  

From now assume that $|A_1| > |A_0|$. We will describe how to compute $A' \subseteq A_1$ and the $S(u)$ sets for each $u \in A'$.
%
%
%We will then continue to Step (ii) described below, which basically just applies the derandomization the basic coloring procedure of Sec. \ref{sec:deg-very-small} using the $S(u)$ sets as a substitute to the original palettes. (To compensate for loosing $99/100$ of the colors, we will reduce the probability that a node picks a color in the basic randomized algorithm of Sec. \ref{sec:deg-very-small}).
Each node $u \in A_1$ sets $S(u)=B_{i,u}$ with probability $|B_{i,u}|/F_u$ for every $i \in \{1,\ldots, \Delta^{1/3}\}$.
The decisions of the nodes are $O(1)$-wise independent. Then, each node $u$ sends to its $A_1$-neighbors the index $i_u$ of its chosen bin. Let $r_{i_u}$ be the total number of $u$'s neighbors in $A_1$ that picks the $i_u$'th bin. We then say that $u$ is \emph{happy} if $r_{i_u}\leq 11|B_{i,u}|$. That is, the number of relevant neighbors of $u$ w.r.t $S(u)$ is at most $11|S(u)|$. Otherwise, the node is sad. 
The final set $A'$ contains all the happy nodes. This completes the description of the first step (of the $j$'th phase).

\paragraph{Step 2.}
In the second step, we apply a single round of Algorithm $\SimpleRandColor$ on the nodes in $A'$ using the $S(u)$ sets as their palettes. The only modification is that we reduce the probability of a node to pick a color to be $1/22$ (rather than $1/2$). That is, each $u \in A'$ first flips a biased coin such that w.p. $21/22$, $c_u=0$; and with the remaining probability, $u$ picks a color uniformly at random in $S(u)$. The nodes decisions are pairwise independent.

\paragraph{Analysis.}
We first show that each phase can be implemented in $O(1)$ number of rounds. In the first step, each node in $A_1$ sends to each of its at most $\Delta$ neighbors, the statistics on the number of free colors in each of its bins, i.e., $|B_{i,u}|$ for every $i \in \{1,\ldots, \Delta^{1/3}\}$. The total number of sent messages is $O(\Delta^{4/3})=O(n)$. Thus this can be done in $O(1)$ rounds. 

We next claim that the number of nodes that got colored in phase $j$ is reduced by a constant factor in expectation, and thus after $O(\log n)$ phases all nodes are colored w.h.p. 
We start by showing the following:
\begin{lemma}\label{cor:final}
With high probability, at the end of the first phase, we have computed a subset of nodes $A'$ and a subset colors $S(u) \subseteq \Pal(u)$ for each $u \in A'$ such that (i) $A'$ contains at least a constant fraction of the uncolored nodes $V_j$ in expectation, (ii) the number of $u$'s relevant neighbors in $A'$ is at most $11|S(u)|$, where a node $v \in \Gamma_j(u) \cap A'$ is a relevant neighbor of $u$ if $S(u) \cap S(u) \neq \emptyset$. 
\end{lemma}
To prove the lemma, we will first consider the more interesting case where $|A_1|\geq |A_0|$. We will show that there are many happy nodes in $A_1$ after the first step.
Let $R_{i,u}$ be the random variable indicating the number of $u$'s neighbors that chose the $i$'th bin, for every $i \in \{1,\ldots, \Delta^{1/3}\}$. We call a bin $B_{i,u}$ \emph{bad} if $|B_{i,u}| \leq \Exp(R_{i,u})/10$, and otherwise it is \emph{good}. Note that the classification into bad and good does not depend on the random choices. 
\begin{claim}\label{cl:clonce}
W.h.p, for every good and large bin $B_{i,u}$, it holds that $R_{i,u}\leq 11 \cdot |B_{i,u}|$.
%B_{i,u}
\end{claim}
\begin{proof}
%Recall that a bin $B_{i,u}$ is good if $|B_{i,u}| \geq 10 \Exp(R_{i,u})$ and it is large if $|B_{i,u}|\geq \Delta^{1/12}$.  
Fix a node $u \in A_1$, and consider a good and large bin $B_{i,u}$ with $\mu=\Exp(R_{i,u})$.
Since $|B_{i,u}|=\Omega(\Delta^{1/12})$ (as it is large) and as $|B_{i,u}|=\Omega(\Exp(R_{i,u}))$ (as it is good), using the concentration inequality for $c$-wise independence from \cite{Vadhan12}, we get that 
$$\Prob(|R_{i,u}-\mu|\geq |B_{i,u}|)\leq ((c\mu+c^2)/|B_{i,u}|^2)^{c}\leq 1/n^{10}.$$
Therefore, with high probability it holds that $R_{i,u}\leq \mu+|B_{i,u}|\leq 11|B_{i,u}|$ . The claim holds by applying the union bound over all $\Delta^{1/3}$ bins of $u$, and over all $n$ nodes.
\end{proof}

\begin{claim}\label{cl:cltwo}
The probability that $u$ chooses a bad bin is at most $1/10$.
%B_{i,u}
\end{claim}
\begin{proof}
Since $\Exp(R_{i,u})=\sum_{v \in \Gamma_j(u)\cap A_1} |B_{i,v}|/F_v$, we have that 
$$\sum_i \Exp(R_{i,u})\leq \sum_i \sum_{v \in \Gamma_j(u)} |B_{i,v}|/F_v =\sum_{v \in \Gamma_j(u)}\sum_i |B_{i,v}|/F_v \leq F_u-1~.$$ 
Therefore,
$$\Prob(u \mbox{ chooses bad bin})=\sum_{i: |B_{i,u}| \leq \Exp(R_{i,u})/10} |B_{i,u}|/F_u \leq \sum_{i: |B_{i,u}| \leq \Exp(R_{i,u})/10}  \Exp(R_{i,u})/(10F_u) \leq 1/10~.$$
\end{proof}

\begin{corollary}\label{cor:happy}
The probability that $u \in A_1$ is happy is at least $1/2$. 
\end{corollary}
\begin{proof}
There are three bad events that prevent $u$ from being happy. 
(i) $u$ chose a bad bin. By Claim \ref{cl:cltwo}, this happens with probability at most $1/10$.\\
(ii) $u$ chose a small bin. Since $u \in A_1$, at least $9/10$ of its colors are in large bins. The probability to choose a small bin is at most $\sum_{i: |B_{i,u}|\leq \Delta^{1/12}}|B_{i,u}|/F_u\leq 1/10$.\\
(iii) $u$ chose a large and good bin $B_{i,u}$, but with at least $r_{i_u}> 11 \cdot |B_{i,u}|$ relevant neighbors.
By Claim \ref{cl:clonce}, this happens with probability at most $1/n^5$. Therefore, the probability that none of the bad events happened is at least $1/2$.
\end{proof}
%We are now completing the proof of Lemma \ref{cor:final}.
\begin{proof}[Proof of Lemma \ref{cor:final}]
First consider the simpler case where $|A_1|\geq |A_0|$. By Cor. \ref{cor:happy}, $|A'|$ contains at least half of the nodes in $A_1$ in expectation. Since each node in $A'$ is happy, claim (ii) holds as well.
Next, consider the complementary case where $A'=A_0$ and for 
every $u \in A_0$, the set $S(u)$ is the union of the colors in the small bins of $u$. By the definition of $A_0$, $|S(u)|\geq 1/10 |\Gamma'(u)| \geq 1/10 |\Gamma(u)\cap A_0|$, where $\Gamma'(u)$ is the set of $u$'s neighbor at the beginning of the $j$'th phase. Claim (ii) holds as well. 
%
%
%
%
 %$|A_1|\geq |A_0|$. Let $A' \subseteq A_1$ be the set of nodes $u$ that picked a good and large bin $B_{i,u}$ such that $|B_{i,u}|\leq R_{i,u}/10$.  
%By the definition of $A_1$, the total number colors in the large bins of  $u \in A_1$ is at least $1/2$ of the total number of its free colors $F_u$. Therefore, with constant probability $u$ chooses a good and large bin $B_{i,u}$ such that $|B_{i,u}|\leq R_{i,u}/10$. 
%To see that, by Claim \ref{cl:clonce}, the probability that $u$ picks some bin $B_{i,u}$ with $R_{i,u}\leq 10\Exp(R_{i,u})$ (i.e., a bin with many neighboring competitors) is at most $1/n^9$. In addition, the probability that $u$ picks a bad bin is at most $1/10$ by Claim \ref{cl:cltwo}. The probability that it chooses a small bin is at most $1/2$. Thus by the union bound, the probability that $u$ chooses a bad bin or a bin with many competitors is at most $5/6$. We get that with constant probability $u$ chooses a large bin such that $|B_{i,u}|\leq R_{i,u}/10$. Since the number of relevant neighbors in $A_1$ is exactly $R_{i,u}$, claim (ii) holds as well. 
\end{proof}
Due to Lemma \ref{cor:final}, we know that $A'$ is large in expectation, and thus it is sufficient to show that Step 2 colors each node $u \in A'$ with constant probability (when using pairwise independence). Assume that $c_u >0$. Then, for any neighbor $v$, w.p. $1/22$ it holds that $c_v\neq 0$. In addition, given that $c_v \neq 0$, the probability that $c_u=c_v$ is $1/|S(u)|$. Over all, $\Prob(c_u =c_v ~\mid~ c_u\neq 0)=1/22 \cdot 1/|S(u)|$. 
Since the number of $u$'s relevant neighbors is at most $11|S(u)|$, by applying the union bound, we get that the probability that $u$'s color is legal, given that $c_u \neq 0$, is $1/2$. The probability that $c_u\neq 0$ is $1/22$, and thus the probability the $u$ is colored is $1/44$.

\paragraph{Derandomization for the First Step.}
In the case where $|A_0|\geq |A_1|$ (i.e., many nodes have lots of colors in small bins), the first step is deterministic. So, we will consider the case where $|A_1|\geq |A_0|$.
Since the arguments are based on $O(1)$-wise independence, the total seed length is $O(\log n)$. By letting each node $v$ send the values of $|B_{i,u}|$ for every $i$, for any possible seed, each node $u$ can simulate the selection of the bins made by its neighbors. Our goal is to maximize the number of happy nodes: nodes $u$ that picked a bin $B_{i,u}$ with low competition, i.e., such that the number of their $A_1$-neighbors that picked that bin is at most $11|B_{i,u}|$. Since in expectation, the number of these nodes is at least a constant fraction of the current uncolored nodes, using the method of conditional expectation, we can compute the desired $O(\log n)$-bit seed within $O(1)$ rounds. 

\paragraph{Derandomization for the Second Step.}
The second step would start by letting each node $u \in A'$ send its set $S(u)$ to each of its relevant neighbors in $A'$. In the case where $A'=A_0$, each node sends $S(u)$ to all its $A'$-neighbors. 
In the other case, for each node $u \in A'$, $S(u)=B_{i,u}$ for some $i \in \{1,\ldots, \Delta^{1/3}\}$. The node $u$ will then send $S(u)$ to any $A'$-neighbor $v$ that also picked its $i^{th}$ bin, i.e., to any neighbor $v \in A'$ with $S(v)=B_{i,v}$. These are the only potential competitors for $u$ on its colors.  

We are now in the situation where each node knows the ``palettes" of its relevant neighbors. 
From that point on the derandomization of the single step of the modified $\SimpleRandColor$ Algorithm is 
basically the same as in Sec. \ref{sec:deg-very-small}. As shown above, when using pairwise independence, a node in $A'$ gets colored w.p. at least $1/44$. As each node $u$ knows the palettes $S(v)$ of its relevant neighbors, using a seed (and the nodes ID), they can simulate the decision of their relevant neighbors. The goal would be to maximize the number of colored nodes using the method of conditional expectation, in the exact same manner as in Sec. \ref{sec:deg-very-small}.

\paragraph{Round Complexity.} Unlike the randomized algorithm, our deterministic solution requires that the nodes will be able to simulate the decisions made by their neighbors. For that purpose we will need to send the set $S(u)$ of colors to all relevant neighbors of $u$. We will show that this is possible both for $A_0$ and for the case where $A' \subseteq A_1$.  

First assume that $|A_0|\geq |A_1|$. The total number of messages sent and received\footnote{As each node in $A_0$ has at most $O(\Delta^{5/12})$ neighbors in $A_0$ by definition.} by a node $u \in A_0$ is $O(\Delta^{2/3+2/12})=O(n)$. So sending all these palettes can be done in $O(1)$ number of rounds. Now consider the case where $|A_1|\geq |A_0|$ and let $A' \subseteq A_1$ be the set of happy nodes. Each node $u \in A'$ sends its chosen bin  $B_{i,u}$ to $O(|B_{i,u})$ relevant neighbors. Since $|B_{i,u}|=O(\Delta^{2/3})$, each node sends at most $O(\Delta^{4/3})=O(n)$ messages. Again, this can be done in $O(1)$ number of rounds. 

\subsection{Handling the General Case} 
We will have a $O(1)$-round procedure to partition the nodes into $\ell=O(\Delta^{1/4})$ sub-graphs $G_1,\ldots, G_{\ell}$ and one additional sub-graph $G^*$ that will be list-colored after coloring the $G_i$ subgraphs. Each subgraph $G_i$ will be assigned a disjoint set of $\Delta(G_i)+1$ colors, and thus we will be able to color all these subgraphs in parallel using $O(\log \Delta)$ rounds using the algorithm of Lemma \ref{lem:det-small-deg}. 

We first show that the desired partitioning can be done with $O(1)$-wise independence and then explain how to derandomize it. Each node picks a subgraph $i \in [1,\ell]$ w.p. $p_i=1/\ell-q/\ell$ for $\ell=\lceil \Delta^{1/4} \rceil$ and $q=2\Delta^{-5/16}$. The decisions of the nodes are $O(1)$-wise independent. 
Using the basic concentration bound for $O(1)$-wise independence, we get that w.h.p., 
$$\Delta(G_i)\leq \Delta \cdot p_i+(\Delta\cdot p_i)^{7/12} \leq \Delta^{3/4}-q \cdot \Delta^{3/4}+\Delta^{7/16}\leq \Delta^{3/4}-1~.$$
Therefore the total number of colors consumed by the $\ell$ subgraphs is at most $\Delta$. 
By using concentration bounds again, w.h.p., the maximum degree of the left over-subgraph $G^*$ is at most $O(\Delta \cdot q)=O(\Delta^{11/16})=O(n^{3/4})$. Thus after coloring the graphs $G_1,\ldots, G_\ell$, the final sub-graph $G^*$ can be colored in $O(\log \Delta)$ rounds by applying again the algorithm of Lemma \ref{lem:det-small-deg}. 

We now show that the random partitioning can be derandomized. Given a seed of $O(\log n)$-bits seed, each vertex $u$ can determine its degree in its chosen subgraph. We have $\ell+1$ subgraphs $G_1, \ldots, G_\ell, G^*$ such that w.h.p. over the $O(\log n)$-bit random seed the maximum degree of each $G_i$ subgraph is at most $\Delta^{3/4}$ and the maximum degree of $G^*$ is $O(\Delta^{11/16})$. For a given seed, we say that the subgraph $G_i$ is bad if its maximum degree exceeds $\Delta^{3/4}$. In the same manner, $G^*$ is bad if its maximum degree exceeds $\Delta^{11/16}$. Our goal is to minimize the number of bad subgraphs. 
%We will allocate a node for each subgraph $i \in \{1,\ldots, \ell, *\}$ and will compute a chunk of $\log n/4$ bits in the seed within $O(1)$ number of rounds. 
For each partial assignment, the node $u$ will compute the probability that its maximum degree violates the desired bound (i.e., $\Delta^{3/4}$ for $G_i$ and $\Delta^{5/6}$ for $G^*$). We will then pick the assignment that minimize the probability (or expectation) of these violations. Since the expected number of bad subgraphs over the random seeds is zero, using this conditional expectation method we will find a seed whose total number of violations is $0$. 

Next, the algorithm applies the procedure of Lemma \ref{lem:det-small-deg} to color the subgraphs $G_1,\ldots, G_\ell$, and then applied the same procedure to list-color the final subgraph $G^*$. We next observe that these subgraphs can be colored simultaneously within the same number of rounds. This clearly holds for the randomized procedure described in Subsec. \ref{sec:threequarter}.  The derandomization procedure is based on revealing a chunk of $\lfloor \log n \rfloor$  bits at at time by sending $O(\log n)$ bits to at most each node in the graph. Since we need to run at most $n^{1/4}$ instances in parallel the derandomization will be slightly changed. Instead of revealing a chunk of $\lfloor \log n \rfloor$ bits, the algorithm of Lemma \ref{lem:det-small-deg} will reveal only $\lfloor \log n/2 \rfloor$ bits in $O(1)$ rounds. Since the seed length is $O(\log n)$, the derandomization will be still done in $O(1)$ rounds. As revealing the values of $\lfloor \log n/2 \rfloor$ bits in the seed requires only $2^{\lfloor \log n/2 \rfloor}\leq n^{1/2}$ nodes (one per assignment to these bits), all the $O(n^{1/4})$ instances can be now derandomized simultaneously. 

\\
\\
\textbf{Acknowledgment:} I am very grateful to Hsin-Hao Su, Eylon Yogev, Seth Pettie, Yi-Jun Chang and the anonymous reviewers for helpful comments. I am also grateful to Philipp Bamberger, Fabian Kuhn and Yannic Maus for noting a missing part in the deterministic $(\Delta+1)$ coloring algorithm of the earlier manuscript.

\bibliographystyle{alpha}
\bibliography{clique}

\appendix

\section{Missing Proofs}\label{apn:missing}
\begin{theorem} (\textbf{Simple Corollary of Chernoff Bound})\label{thm:cher} Suppose $X_1$, $X_2$, \dots, $X_\ell \in [0,1]$ are independent random variables, and let $X=\sum_{i=1}^{\ell} X_i$ and $\mu = \mathbb{E}[X]$. If $\mu \geq 5 \log n$, then w.h.p. $X \in \mu \pm \sqrt{5\mu\log n}$, and if $\mu < 5 \log n$, then w.h.p. $X \leq \mu +5\log n$.
\end{theorem}

\subsection{Missing Algorithms from \cite{CHP18}}\label{sec:missingalg}

For completeness, we provide the psuedocodes for coloring procedures from \cite{CHP18} that are used in our algorithm as well.
Define $N^*(v)=\{u \in N(v) ~\mid~ ID(u)<ID(v)\}$. 
\begin{mdframed}[hidealllines=false]
$\OneShotColoring$. Each uncolored vertex $v$ decided to participates independently with probability $p$. Each participating vertex $v$ selects a color $c(v)$ from its palette $\Psi(v)$ uniformly at random. A participating vertex $v$ successfully colors itself if $c(v)$ is not chosen by any vertex in $N^*(v)$.
\end{mdframed}

In the $\ColorBidding$ procedure each vertex $v$ is associated with a parameter $p_v \geq |\Psi(v)|-\outdeg(v)$ and $p^*=\min_v p_v$. Let $C$ be a constant satisfying that $\sum_{u \in N_{out}(v)}1/p_u\leq 1/C$. All vertices agree on the value $C$.
\begin{mdframed}[hidealllines=false]
$\ColorBidding$. Each color $c \in \Psi(v)$ is added to $S_v$ with probability $C/2p_v$ independently. If there exists a color $c^* \in S_v$ that is not selected by al vertices in $N_{out}(v)$, $v$ colors itself $c^*$.
\end{mdframed}

\subsection{Implementing CLP for $\Delta=O(\sqrt{n})$ (\Cref{thm:colorsqrtn})} \label{sec:clpdetails}
For psuedocodes from \cite{CHP18}, see \Cref{sec:missingalg}.\\
\textbf{Proof of \Cref{thm:colorsqrtn}}:
\APPENDFULLCLPSMALL

\subsection{CLP with Non-Equal Palette Sizes}\label{sec:modclp}
In this section, we provide a detailed proof for \Cref{lem:modclp}.
\APPENDUNBAL

%\section{Deterministic Coloring}\label{append:det}
%%\input{prelim}

\end{document}